\documentclass{article}

\usepackage{dilab_arxiv}

\usepackage{amsfonts}
\usepackage{mathtools}
\usepackage{array}
\usepackage{enumitem}
\setlist[itemize]{itemsep=1pt, topsep=2pt, parsep=0pt, leftmargin=1.2em}
\usepackage{multirow}
\usepackage{wrapfig}
\usepackage{algorithm}
\usepackage{algorithmic}
\usepackage{placeins}
\newtheorem{theorem}{Theorem}
\newtheorem{proposition}[theorem]{Proposition}
\newtheorem{lemma}[theorem]{Lemma}
\newtheorem{corollary}[theorem]{Corollary}

\newtheorem{remark}[theorem]{Remark}

\newcommand{\eps}{\varepsilon}
\newcommand{\cX}{\mathcal{X}}
\newcommand{\cI}{\mathcal{I}}
\newcommand{\cZ}{\mathcal{Z}}
\newcommand{\xbar}{\bar{x}}
\newcommand{\sbar}{\bar{\sigma}}
\newcommand{\dbar}{\bar{d}}
\newcommand{\mubar}{\bar{\mu}}
\newcommand{\Rbar}{\bar{R}}
\newcommand{\Rhat}{\widehat{R}}
\newcommand{\gap}{\mathrm{NashGap}}
\newcommand{\KL}{\mathrm{KL}}
\newcommand{\ip}[2]{\langle #1, #2 \rangle}

\newlength{\figbase}
\title{When Are We Near Nash?\\Certifying Solutions in General-Sum Games}
\runningtitle{Certifying Solutions in General-Sum Games}
\date{\today}

\paperlogo{\includegraphics[height=1.5cm]{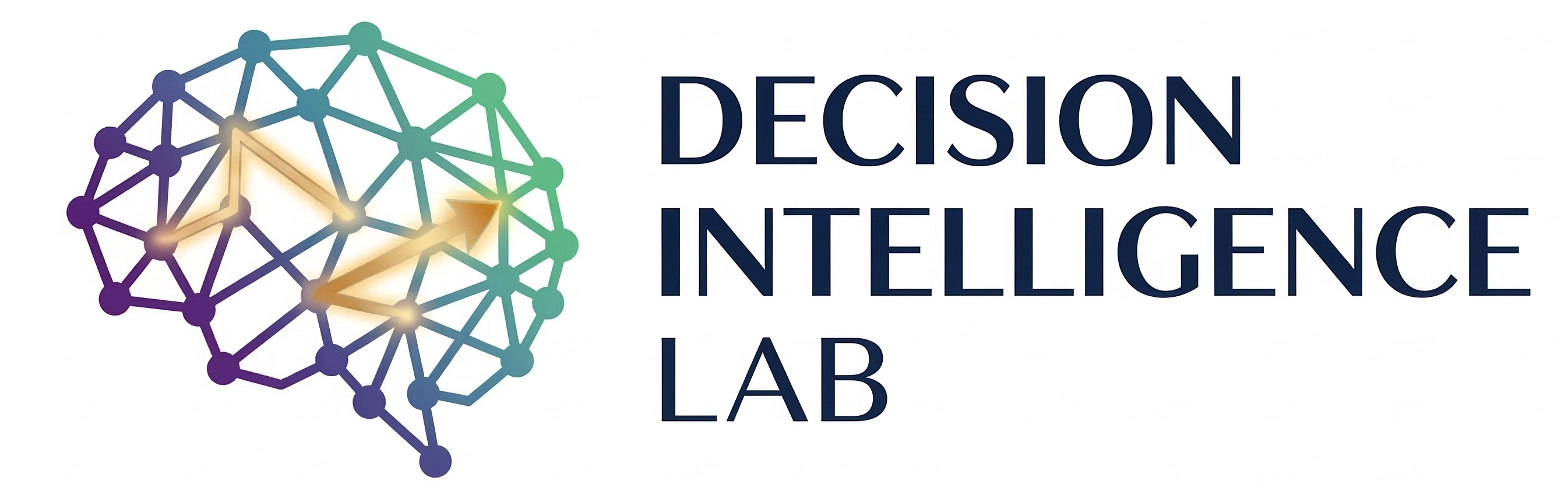}}

\author{
  Boning Li$^{1}$ and Longbo Huang$^{1\,\text{\faEnvelope}}$
  \\[0.3em]\normalfont
  $^1$Institute for Interdisciplinary Information Sciences, Tsinghua University
  \\
  \text{\faEnvelope}\ Correspondence: longbohuang@tsinghua.edu.cn
}

\begin{document}
\maketitle
\thispagestyle{fancy}

\begin{abstract}
Heuristic solvers for general-sum extensive-form games produce a stream of
iterates, and practice reports the NashGap of the returned profile.
A recent meta-learning benchmark reports the \emph{lowest} NashGap along the
training trajectory, located by exact evaluation at every checkpoint, without
stating the selection rule or the evaluation budget behind that number. We ask
what a solver can know about its own iterates from the regret accumulators and
the terminal-distribution accumulators its updates already maintain.
In two-player games the NashGap of a weighted average strategy
\emph{equals}, at the maximizing player, the weighted external regret plus a
signed value correction.
Bounding the regret by
the running counterfactual-regret accumulator gives an online certificate
whose cost per checkpoint is linear in the number of terminal histories.
The certificate stays sound under alternating updates and in games with three
or more players. The mismatch between what
alternation evaluates and what the averages record enters as one signed
comparator term, maximized in a single backward pass per player whose
comparator weight is nonzero. The bound is
read off the regret state, with no best response against the average. Run a
portfolio of training-free regret-matching solvers, output the iterate with the
smallest certificate, and verify that iterate once.
On the benchmark's biased Leduc family a pool of sixteen training-free solvers,
seven original and nine recent variants, solves $89\%$ of the $64$ instances
at threshold $10^{-5}$ this way, matching the trajectory-minimum union of the
solvers. The same protocol solves $92\%$ of the $64$ pre-registered held-out
instances and $89\%$ again on a third set of $64$ instances that no earlier
stage used. The zero-best-response certificate certifies all $57$ in-sample
solves and all $57$ fresh ones, and $58$ of the $59$ held-out ones. Its
exact-comparator refinement certifies every solve in the three sets. At the
selected iterate the refinement tracks the true gap within a median factor of
$1.4$ over the $36$ instances whose selected gap is above $10^{-12}$. The
meta-trained method reports $96\%$ under the trajectory-minimum protocol.
That protocol leaves its selection rule and evaluation budget unstated. The
certificate states both and attributes its bound. Code is available at
\url{https://github.com/lbn187/CertNE}.
\end{abstract}

\section{Introduction}
\label{sec:intro}

Computing a Nash equilibrium of a general-sum game to precision
inverse-polynomial in the game size is PPAD-complete (Polynomial Parity
Arguments on Directed graphs), for two players
\citep{chen2009settling} and for three or more players
\citep{daskalakis2009complexity}. The no-regret dynamics
that solve two-player zero-sum games at scale
\citep{zinkevich2007regret,tammelin2014solving,li2026gpu,li2026real}
guarantee only
coarse correlated equilibria in general \citep{anagnostides2022near}. Practice
therefore runs heuristic solvers and judges them afterwards by their
NashGap, the largest gain any player could obtain by deviating alone
\citep{lisy2017eqilibrium,davis2014using}.

Computing that gap exactly costs one best-response pass per player over the
game tree, the same order as one iteration for a tabular solver
(Table~\ref{tab:budget}). The certificate below builds a bound from the
solver's own regret state, without a best response. Under aligned two-player
play, where every accumulator refers to the same start-of-round profile, an
exact evaluator that also records the running realized value recovers the
same regret and value terms from its best-response values. Under alternation,
or with three or more players, the certificate adds a comparator term, the
signed functional of Section~\ref{sec:update-order}, which an exact evaluator
does not compute. With that term the raw regret state is a bound again.

A recent benchmark shows why the scoring protocol matters. Neural predictive
counterfactual regret minimization$^+$ (NPCFR$^+$)
\citep{sychrovsky2026approximating} meta-learns a predictive bias. On
biased two-player Leduc games it reports solving $96\%$ of instances to
NashGap $10^{-5}$; its untrained counterpart PCFR$^+$ reports $16\%$. The
metric is the \emph{lowest} NashGap along the training trajectory, located
by exact evaluation at roughly $19$ checkpoints over $2^{18}$ iterations. The
reported number is therefore the output of a selection rule, the argmin of
the exact gap, and an evaluation budget of one gap per checkpoint. Neither
is part of the method or stated with the result
(Section~\ref{sec:oracle-gap}). The quantity the same paper proposes for
bounding the gap without a best response combines a regret term with a
Pinsker-type correlation term. It is the natural candidate for such a rule,
and we find it sits three orders of magnitude above the true gap and grows.

We give an online NashGap certificate computed from the regret accumulators
and the terminal accumulators, together with an
account of which quantities the gap actually depends on. In two-player
games the correlation term is avoidable. The comparator term is
exact by bilinearity, and the value term becomes exact once the
distribution exchange is replaced by directly measurable quantities. The
inequality therefore becomes an \emph{identity} (Theorem~\ref{thm:tight}).
The NashGap of the weighted average strategy \emph{equals} the weighted
external regret plus a signed value correction. Bounding the regret by the
running counterfactual-regret accumulator (Lemma~\ref{lem:weighted-decomp})
gives a certificate whose per-checkpoint cost is linear in the number of
terminal histories.

The strongest solvers update players in alternation
\citep{tammelin2014solving,brown2019solving,burch2019revisiting}. The regrets they
evaluate then refer to a partially updated opponent, while the averages
record the start-of-round profile. A second regret pass per iteration is
the usual remedy. The mismatch itself is one more measurable comparator
term, a signed linear functional maximized in one backward pass per player
with nonzero comparator weight, so no second regret pass is needed. The
certificate is therefore sound for any update
order and any number of players (Proposition~\ref{prop:sequence}). In
games with three or more players the same term is the coalition imbalance
that the correlation level is structurally blind to
(Remark~\ref{rem:eps-blind}). No bound on the coalition term vanishes with
the correlation level, so the Pinsker-type route fails beyond two players;
the comparator term measures it directly. The
certificate splits its bound into a regret term, a value correction and
comparator terms (Table~\ref{tab:guarantees}).

The certificate turns the post-hoc protocol into an algorithm
(Figure~\ref{fig:protocol}). Run a portfolio of training-free
regret-matching solvers, weighted, predictive, discounted, scheduled and
stabilised variants
\citep{tammelin2014solving,farina2021faster,brown2019solving,xu2024minimizing,zhang2026faster,farina2023regret,zhang2025scale},
and record the certificate at every checkpoint of every solver. Output the
iterate with the smallest certificate, so the selection rule and its budget
are explicit. The zero-best-response certificate ($C^0$) reads accumulators
alone, and its exact-comparator refinement ($C^1$) adds one signed backward
pass per player with nonzero comparator weight at each checkpoint. On biased
Leduc this is a single pass, costing about as much wall-clock as an exact gap
evaluation and less than one training iteration. The selected iterate can be
verified exactly in one further gap evaluation (Corollary~\ref{cor:one-shot}).

\begin{figure}[t]
\centering
\includegraphics[width=\textwidth]{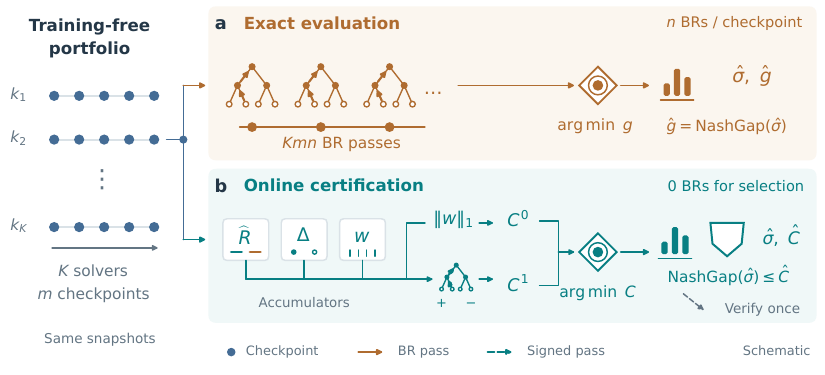}
\caption{Two selection protocols on the same $K$ solvers and $m$
checkpoints per solver. \textbf{a,} Exact evaluation scores every candidate
with $n$ best-response (BR) passes, costing $Kmn$ passes overall, and
outputs the minimum-gap strategy with its exact gap.
\textbf{b,} Online certification combines signed regret, value correction,
and comparator terms (Proposition~\ref{prop:sequence}). $C^0$ uses
accumulators alone; $C^1$ adds one signed backward pass per player with
nonzero comparator weight at each checkpoint. Either certificate selects
its own argmin across solvers and checkpoints and bounds the selected
strategy's NashGap. Optional verification evaluates only that strategy,
with $n$ BR passes (Corollary~\ref{cor:one-shot}). Nodes, trees and policy
glyphs are schematic; arrows on trees indicate backward traversal.}
\label{fig:protocol}
\end{figure}

On the benchmarks of \citet{sychrovsky2026approximating} this protocol
matches the trajectory-minimum rule with no exact-gap query during
selection; the cost it does pay is the accumulator bookkeeping
(Table~\ref{tab:budget}). On biased
Leduc ($64$ instances), certificate selection over sixteen training-free
solvers solves $89\%$ of instances at threshold $10^{-5}$, the same fraction
as the trajectory-minimum union over the same solvers. The seven original
solvers alone reach $84\%$, and the sixteen-solver
pool, extended by nine recent variants in a second pre-registered stage,
reaches $92\%$ on $64$ held-out instances. On a third instance set that no earlier stage
used, the protocol repeats at $89\%$ ($57/64$).

The certificate itself is below the threshold on every in-sample solve, for
$C^0$ read off accumulators alone and for $C^1$ with one signed backward
pass per checkpoint. Across the three sets $C^1$ certifies every solve, and
$C^0$ all but one held-out solve. At the selected iterate $C^1$ is within a
median $1.4\times$ of the true gap over the $36$ instances whose selected gap
is above $10^{-12}$. The aligned start-of-round certificate pays an extra
regret pass per player per iteration yet is looser ($1.9\times$) and
certifies $86\%$. On the single deterministic three-player Leduc game,
PDCFR$^+$ reaches NashGap $1.17\times10^{-6}$ while the certificate-selected
HS-PCFR$^+$ run is certified at $2.83\times10^{-6}$, both below the reported
$10^{-3}$ NPCFR$^+$ comparator. Under our protocol the training-free
portfolio reaches $0.89$ against the reported $0.96$. The certificate is
solver-agnostic and attaches to NPCFR$^+$ or any other update sequence
unchanged. Our $7$ in-sample failures cluster in games with a low
tie-refund fraction (Section~\ref{app:hard-instances}).

Our contributions are:
\begin{itemize}
  \item \textbf{Theory.} We prove that in two-player games the NashGap of a
    weighted average \emph{equals} its weighted external regret plus a
    signed value correction, an identity, not a bound
    (Theorem~\ref{thm:tight}). We build from it an online certificate whose
    zero-best-response form reads accumulators alone and whose exact form
    costs one signed backward pass per player with nonzero comparator
    weight, sound for any update order and player count
    (Proposition~\ref{prop:sequence}). We show that in games with three or
    more players no correlation level controls the comparator term
    (Remark~\ref{rem:eps-blind}), and the certificate measures it
    directly. One exact gap evaluation verifies the selected iterate
    (Corollary~\ref{cor:one-shot}).
  \item \textbf{Algorithm.} We turn the certificate into the selection
    rule: score every checkpoint of every solver, output the argmin, and
    report the candidate with a stated budget and no exact-gap query
    during selection. The argmin comes from accumulators alone under
    $C^0$; under $C^1$ it costs one signed backward pass per player with
    nonzero comparator weight.
  \item \textbf{Re-evaluation under a stated selection rule.} On the NPCFR
    benchmarks, we show that training-free certificate selection reaches
    $89\%$ in-sample, matching the trajectory-minimum union over the same
    solvers, $92\%$ held-out and $89\%$ on a third, never-used set,
    against the reported meta-trained $96\%$. $C^1$ certifies every solve
    in all three sets, and $C^0$ certifies all but one held-out solve.
    The residual in-sample failures show the same lean.
\end{itemize}

\section{Related Work}
\label{sec:related}

\textbf{Meta-learning and steering.}
Learning-not-to-regret meta-learns regret-minimizer internals for faster
convergence \citep{sychrovsky2024learning}, following the general
meta-learning-in-games program of \citet{harrismeta}. NPCFR
\citep{sychrovsky2026approximating} redirects this machinery at
general-sum Nash approximation by meta-learning a predictive bias, and is
the reference point for our re-evaluation;
\citet{sychrovsky2025selfplay} extends the line to self-play. An
alternative channel steers no-regret learners toward equilibria through
external payments \citep{zhang2024steering}. Our core methods RM$^+$ and
PRM$^+$ are standard training-free algorithms; DPRM$^+$ is also training-free
but modifies the PRM$^+$ prediction with an analytic decorrelation term.
The portfolio-level selection among their iterates is new.

\textbf{Certificates and gap evaluation.}
Theorem~1 of \citet{sychrovsky2026approximating} bounds the NashGap by a
regret term plus a term in the extensive-form marginalizability of the
empirical play distribution, via a Pinsker-type inequality, and states it
as an analysis tool. Our
Proposition~\ref{prop:kl-cert} is its weighted-averaging analogue, and
Theorem~\ref{thm:tight} removes the distribution inequality entirely in
the two-player case. Outside these structural certificates, obtaining a NashGap means
computing it, one exact best-response pass per player. Small equilibrium
certificates exploit tree structure to compress the strategy witness
in large games
\citep{zhang2020small}, and accelerating the best-response computation
itself has long been an engineering goal there
\citep{johanson2011accelerating,li2026agents}. Exploitability is likewise the
standard for judging deployed agents
\citep{davis2014using,lisy2017eqilibrium,lockhart2019computing,balduzzi2018re,li2026av}.

\textbf{Regret-based certificates.}
They survive only through zero-sum structure. There the
marginals of a coarse correlated equilibrium are $2\varepsilon$-Nash
\citep{macqueen2023cce}, and the average's gap is bounded by total regret at
every step \citep{zinkevich2007regret}. Since $u_2=-u_1$ forces
$\Delta_1=-\Delta_2$, that bound is the zero-sum case of
Theorem~\ref{thm:tight}, in which the two signed value corrections cancel.
In general-sum play they do not, and the theorem measures the residual from
the running realized value. Faster regret rates, such as those of the $+$ and
predictive variants \citep{tammelin2014solving,farina2021faster}, shrink the
regret term and leave that residual untouched. AIVAT \citep{burch2018aivat} estimates a fixed
agent's value without bias, and anytime PSRO \citep{mcaleer2022anytime} gates a growing
policy population on restricted-game exploitability.
None of these gives an online, best-response-free certificate
for iterate selection during training in general-sum extensive-form games.

\textbf{Regret minimization and equilibria.}
\label{app:extra-related}
Counterfactual regret minimization (CFR) decomposes external regret over
infosets \citep{zinkevich2007regret}. CFR$^+$
\citep{tammelin2014solving,tammelin2015solving} with alternating updates
\citep{burch2019revisiting,li2026alternation} and linear averaging solved heads-up limit
poker \citep{bowling2015heads}. Predictive regret matching$^+$
\citep{farina2021faster}, discounted regrets and weights
\citep{brown2019solving}, their predictive combination PDCFR$^+$
\citep{xu2024minimizing}, hyperparameter schedules for the discount and
weight exponents \citep{zhang2026faster}, the smooth, stable and
extragradient stabilisations of regret matching$^+$ \citep{farina2023regret}
and scale-invariant IR-PRM$^+$ \citep{zhang2025scale} are training-free
accelerants, and one configuration of each is in our pool. The last-iterate
linear convergence of the scale-invariant step is analyzed by
\citet{li2026last}. All of
these guarantees concern regret, which
in general-sum games yields only (coarse) correlated equilibria
\citep{anagnostides2022near,celli2019learning}. Even averaging with an
explicit approximation guarantee, as under fictitious play
\citep{conitzer2009approximation}, works only in restricted settings. Monte Carlo variants
\citep{lanctot2009monte,li2026correlated} trade these exact terminal
distributions for sampling. Abstraction and pruning
\citep{li2024rl,li2025efficient,li2026abstraction,li2026effective} are a
complementary route to scale in
large games (Appendix~\ref{app:future-work}).

\textbf{Portfolios and algorithm selection.}
Running several solvers and selecting among their outputs is classical in
algorithm selection \citep{rice1976algorithm,gomes2001algorithm,xu2008satzilla}.
Here the selection signal is a certificate with a soundness guarantee,
so the selected iterate comes with
a bound on its NashGap.
The task differs from optimizing a \emph{strategy} portfolio, which is already NP-hard
in zero-sum games \citep{drabent2025understanding}, and from solver-dependent selection
among equal-value equilibria \citep{leal2026which}. Ours is a portfolio of no-regret
solvers, and the chosen iterate is scored by the certificate.

\section{Preliminaries}
\label{sec:prelim}

\textbf{Extensive-form games and sequence form.}
We consider finite $n$-player perfect-recall extensive-form games (EFGs).
Each has terminal set $\cZ$, information sets (infosets) $\cI_i$ that group each
player's decision points by what that player observes there, at most $A$
actions per infoset, and fixed chance reach probability $p_c(z)$. Utilities are
$u_i \colon \cZ \to [-M, M]$, where
$M \coloneqq \max_{i,z} |u_i(z)|$. In the sequence-form
representation \citep{koller1994fast,koller1996efficient},
player $i$'s strategy is a realization plan $x_i$ in a convex polytope
$\cX_i$; we write $x_i(z) \in [0,1]$ for the realization weight of player
$i$'s last sequence on the path to $z$. A profile
$\sigma = (x_1, \dots, x_n)$ induces the terminal distribution
\begin{equation}
  d_\sigma(z) \;=\; p_c(z) \prod_{i=1}^n x_i(z),
  \qquad
  u_i(\sigma) \;=\; \ip{d_\sigma}{u_i}
  \;=\; \textstyle\sum_{z \in \cZ} d_\sigma(z)\, u_i(z).
\end{equation}
Here and below $\ip{\cdot}{\cdot}$ is the Euclidean inner product on
$\mathbb{R}^{|\cZ|}$. The map $x_i \mapsto u_i(x_1,\dots,x_n)$ is linear in
each player's realization plan separately, by multilinearity.

\textbf{NashGap.}
For a profile $\sigma$,
\begin{equation}
  \gap(\sigma) \;=\; \max_{i} \; \Big[
    \max_{x' \in \cX_i} u_i\big(x', \sigma_{-i}\big) - u_i(\sigma)
  \Big],
\end{equation}
so $\sigma$ is an $\eps$-Nash equilibrium iff $\gap(\sigma) \le \eps$.
Evaluating $\gap$ exactly requires one best-response (BR) computation per
player, i.e.\ $n$ backward passes over the game tree. Computing an
$\eps$-Nash equilibrium of a general-sum game, for $\eps$ inverse-polynomial
in the game size, is PPAD-complete
\citep{chen2009settling,daskalakis2009complexity}. Regret-minimization
dynamics guarantee convergence only to coarse correlated equilibria, joint
play distributions that no player improves on by deviating before play
begins, in general \citep{anagnostides2022near}.

\textbf{Weighted iterates and online accumulators.}
An iterative solver produces profiles
$\sigma^t = (x_1^t, \dots, x_n^t)$, $t = 1, \dots, T$. Fix nonnegative
weights $w_t \ge 0$ with $W_T \coloneqq \sum_{t \le T}w_t>0$ (we use
$w_t=t^p$, with $p=0$ uniform, $p=1$ linear and $p=2$ quadratic
\citep{tammelin2014solving,brown2019solving}), and let
$\rho_t = w_t / W_T$. Define the following quantities.
\begin{itemize}
  \item the \emph{weighted average strategy}
    $\xbar_i = \sum_t \rho_t x_i^t \in \cX_i$, a realization plan by
    convexity of $\cX_i$. Normalizing per infoset gives the behavioral
    profile $\sbar$, whose terminal distribution
    $d_{\sbar} = p_c \prod_i \xbar_i \eqqcolon \mubar$ equals that of
    $\xbar$ by realization equivalence. The behavioral profile and the
    realization plan are distinct objects whenever player $i$ acts at
    multiple infosets on a path, but they induce the same terminal
    distribution.
  \item the \emph{weighted empirical terminal distribution}
    $\dbar(z) = \sum_t \rho_t\, p_c(z) \prod_i x_i^t(z)$, i.e.\ the
    distribution of play when a round $t \sim \rho$ is drawn and all
    players jointly use $\sigma^t$;
  \item the \emph{weighted external regret}
    $\Rbar_i = \max_{x' \in \cX_i} \sum_t \rho_t
      \big[ u_i(x', \sigma^t_{-i}) - u_i(\sigma^t) \big]$;
  \item the \emph{correlation level}, measured by the Kullback--Leibler
    divergence
    $\eps \,\coloneqq\, \KL\!\big(\dbar \,\Vert\, \mubar\big)$ between the
    empirical joint terminal distribution and the independent-play
    distribution $\mubar = p_c \prod_i \xbar_i$ induced by the average
    realization plans. It is
    the weighted analogue of $\eps$-extensive-form marginalizability
    ($\eps$-EFM) in \citet{sychrovsky2026approximating}, finite since
    $\dbar \ll \mubar$.
\end{itemize}
\textbf{Counterfactual values and regrets.}
For an infoset $I \in \cI_i$ and action $a \in A(I)$, the
\emph{counterfactual action value} at iteration $t$ is
\begin{equation}
  \label{eq:cfv}
  v_i^t(I,a) \;=\; \sum_{z \in \cZ(I,a)} p_c(z)\, x_{-i}^t(z)\,
    \pi_i^t(z \mid Ia)\, u_i(z),
\end{equation}
where $\cZ(I,a) \subseteq \cZ$ collects the terminals reached through $I$
after taking $a$, $x_{-i}^t(z) = \prod_{j \ne i} x_j^t(z)$ is the
opponents' realization weight on the path to $z$, and
$\pi_i^t(z \mid Ia)$ is the product of player $i$'s behavioral
probabilities from $Ia$ to $z$. The counterfactual convention therefore
excludes own reach \emph{before} $I$ and includes own continuation
\emph{after} the action, as in
\citet[Theorem~3]{zinkevich2007regret}. The own reach
$\pi_i^{x'}(I) \in [0,1]$ of a realization plan $x'$ is the product of its
behavioral probabilities along the path to $I$. The \emph{instantaneous
counterfactual regret} of $a$ at $I$ against the current play is
$
  ir_i^t(I,a) \;=\; v_i^t(I,a)
  - \sum_{a' \in A(I)} \sigma^t(a' \mid I)\, v_i^t(I,a').
$

These are byproducts of a counterfactual-regret-minimization (CFR) style
loop \citep{zinkevich2007regret}. The accumulators $\xbar$ and
$\dbar$ and the cumulative counterfactual regrets are running sums updated
in $O(|\cZ|)$
per iteration, and $\Rbar_i$ is bounded online without a best response
(Lemma~\ref{lem:weighted-decomp}). Averaging destroys independence across
players, so $\dbar \ne \mubar$ in general, and that distinction separates
Nash from coarse correlated play \citep{anagnostides2022near}.

\textbf{Regret matching variants.}
A $+$ superscript denotes the regret-matching$^+$ step
\citep{tammelin2014solving}, which floors the accumulated regret at zero
before matching. We reserve CFR$^+$/PCFR$^+$/NPCFR$^+$ for the reference
study's nomenclature, where PCFR$^+$ is predictive CFR$^+$ and NPCFR$^+$
its meta-learned counterpart. Our solvers are RM$^+$, the core of CFR$^+$
\citep{hart2000simple,tammelin2014solving,tammelin2015solving}, PRM$^+$,
predictive RM$^+$ \citep{farina2021faster}, and DPRM$^+$, PRM$^+$ with the
analytic
decorrelation term of Appendix~\ref{app:impl}, all with simultaneous or
alternating updates \citep{burch2019revisiting} and weight exponent
$p \in \{0, 1, 2\}$. The nine solvers added in the second stage keep their
published names. DCFR \citep{brown2019solving}, DCFR$^+$ and PDCFR$^+$
\citep{xu2024minimizing}, HS-PCFR$^+$ and HS-DCFR \citep{zhang2026faster},
Smooth PRM$^+$, Stable PRM$^+$ and ExRM$^+$ \citep{farina2023regret}, and
IR-PRM$^+$ \citep{zhang2025scale} are training-free, as are the original
seven.

\section{The Oracle Gap in Trajectory-Minimum Evaluation}
\label{sec:oracle-gap}

\textbf{The protocol and why it matters.}
The NPCFR benchmark studied here reports the \emph{trajectory minimum}:
given a checkpoint grid $t_1 < t_2 < \dots < t_m$, the reported score of
a run is
\begin{equation}
  \label{eq:traj-min}
  g^\star \;=\; \min_{k \le m} \; \gap\big(\sbar^{\,t_k}\big),
\end{equation}
and a game counts as ``solved'' if $g^\star \le 10^{-5}$
\citep{sychrovsky2026approximating}. The original study uses checkpoints
at $2^k$ (about $19$ evaluations through $T=2^{18}$ on biased Leduc); our
experiments use a $128$-point dense grid for selection and diagnostics. Locating either grid's minimizer requires an exact gap
evaluation ($n$ best-response passes) at every checkpoint, so $g^\star$
is the output of a selection rule, a running argmin of the exact gap,
that is implementable online but is neither part of the method nor stated
with the result. Two consequences
follow. First, \emph{post-hoc selection}. The directly available outputs of a run
are the final average or an iterate chosen by an online rule. For
plus-family algorithms the trajectory of $\gap(\sbar^{\,t})$ is not
monotone (on the benchmark of Section~\ref{sec:experiments}, our
PRM$^+$ control routinely reaches its minimum and then rebounds by an order of
magnitude), so $g^\star$ can be far below the final average a run
returns. Obtaining it takes the same-grid exact selector of
Section~\ref{sec:experiments}. Second, \emph{misalignment}. The protocol compares trajectory
optima, which a run yields only by spending $m$ exact gap queries and which
its reported output, the final average, does not attain. Any conclusion of
the form ``method A solves $X\%$ of games'' silently credits every method
with those queries.

\textbf{What a certificate should do.}
We require the selection rule to be \emph{algorithmic}: an online quantity
$C_t$ computable from the run itself such that
(i)~\emph{soundness}: $C_t \ge \gap(\sbar^{\,t})$, so certifying
$C_t \le \eps$ proves the iterate is an $\eps$-Nash equilibrium;
(ii)~\emph{no best response against the reported profile}: $C^0$ reads
the regret accumulators and the terminal-distribution accumulators kept
alongside the updates, and $C^1$ adds one signed
backward pass per player with nonzero weight, so a checkpoint of either is a
fraction of one training
iteration; (iii)~\emph{tightness}: $C_t$ tracks
$\gap(\sbar^{\,t})$ within a small factor on the trajectories at hand, so that
$\arg\min_t C_t$ preserves the solved-instance decisions of the exact
selector at the thresholds tested; and
(iv)~\emph{comparability}: certificates of different runs are mutually
comparable, enabling portfolio-level selection. Section~\ref{sec:certificates}
constructs such a certificate and shows that requirement (iii) fails for the
natural Pinsker-based candidate, whose correlation term sits three orders
of magnitude above the true gap. It holds empirically for the exact
reformulation. The certificate tracks the gap within a median $1.27\times$
over the $127{,}849$ biased-Leduc checkpoints whose gap is above $10^{-12}$,
and within $1.4\times$ at the selected
iterate over the $36$ instances whose selected gap is above $10^{-12}$.
Soundness (i) is the theorem.

\section{Online Certificates from an Exact Decomposition}
\label{sec:certificates}

All statements below are made only at checkpoints with $W_T>0$; throughout,
a bound $B \ge \gap(\sbar)$ is \emph{sound}, and it is \emph{tight to a
factor $c$} when $B \le c\,\gap(\sbar)$. No multiplicative tightness holds
for all games (Appendix~\ref{app:no-universal-tightness}). Then $\sbar$, $\dbar$, $\mubar$, $\Rbar_i$, and
$\eps = \KL(\dbar \Vert \mubar)$ are the normalized weighted quantities of
Section~\ref{sec:prelim}, accumulated through $T$, and we suppress $T$ in
the notation.
Proofs not given inline are in Appendix~\ref{app:proofs}.
All soundness statements are in exact arithmetic. The experiments evaluate
the bounds in double precision without outward rounding, so a ``certified''
count means that the computed bound is below the threshold;
Section~\ref{sec:experiments} reports the smallest computed certificate-minus-gap
residual and the smallest threshold margin on each instance set.

The weighted external regret $\Rbar_i$ involves a maximization over
$\cX_i$ and is not observable without a best response; the
counterfactual-regret decomposition bounds it online.

\begin{lemma}[Weighted regret decomposition]
\label{lem:weighted-decomp}
Let $v_i^t(I,a)$ and $ir_i^t(I,a)$ be the counterfactual action value and
the instantaneous counterfactual regret of Section~\ref{sec:prelim}. Then
\begin{equation}
  \Rbar_i \;\le\; \Rhat_i \;\coloneqq\;
  \frac{1}{W_T} \sum_{I \in \cI_i}
    \Big[ \max_{a} \sum_{t \le T} w_t \, ir_i^t(I,a) \Big]^{\!+},
\end{equation}
where $[x]^+ \coloneqq \max\{x,0\}$.
\end{lemma}

The proof in Appendix~\ref{app:proofs} applies the laminar expansion of
\citet[Theorem~3]{zinkevich2007regret} per iteration, sums with the weights
$w_t$, and only then takes the per-infoset maximum and positive part.
$\Rhat_i$ uses a raw signed accumulator $Q_i(I,a)=\sum_{t\le T}w_t
ir_i^t(I,a)$ of the regrets the update evaluates, kept apart from the solver's
clipped state; the lemma holds for whatever opponent sequence those regrets
were measured against.

\begin{lemma}[Distribution exchange]
\label{lem:exchange}
For any $f \colon \cZ \to [-M, M]$ and terminal measures $p, q$,
$\;|\ip{p - q}{f}| \le M \, \lVert p - q \rVert_1$; if $p, q$ are
probability distributions then
$\lVert p - q \rVert_1 \le \sqrt{2 \, \KL(p \Vert q)}$ (Pinsker; \citet{cover1991elements}).
\end{lemma}

\textbf{The Pinsker certificate and its bottleneck.}
\label{sec:pinsker-cert}

The natural certificate is the weighted analogue of the
$\eps$-EFM-based bound of \citet{sychrovsky2026approximating}. It combines
the two lemmas.

\begin{proposition}[Pinsker certificate, $n = 2$]
\label{prop:kl-cert}
For two-player games and any nonnegative weights,
\begin{equation}
  \gap(\sbar) \;\le\; \max_i \Rbar_i \;+\; M\sqrt{2\eps}.
\end{equation}
\end{proposition}

The proof in Appendix~\ref{app:proofs} handles the comparator term
exactly by bilinearity and relaxes only the value term through
Lemma~\ref{lem:exchange}. Its
deployable form replaces $\Rbar_i$ by $\Rhat_i$
(Lemma~\ref{lem:weighted-decomp}).

Proposition~\ref{prop:kl-cert} is sound but overestimates the true NashGap
by roughly three orders of magnitude on our benchmark (median
$1.3\times10^3\times$). The regret term is $\sim\!1\times$ tight; the
distribution-exchange relaxation dominates, since late in training $\eps \approx 10^{-9}$
yet $\sqrt{2\eps} \approx 10^{-4.5}$, orders above a gap of
$10^{-6}$--$10^{-8}$.

\textbf{An exact decomposition for two players.}
\label{sec:tight-cert}

The proof of Proposition~\ref{prop:kl-cert} relaxes only the value term.
That term needs no relaxation either: it is a signed, directly measurable
quantity, and saying so turns the inequality into an identity.

\begin{theorem}[Exact NashGap decomposition, $n = 2$]
\label{thm:tight}
For two-player games and any nonnegative weights,
\begin{equation}
  \label{eq:tight-identity}
  \gap(\sbar) \;=\; \max_i \big[\Rbar_i + \Delta_i\big],
  \qquad
  \Delta_i \;\coloneqq\; \ip{\dbar}{u_i} - u_i(\sbar).
\end{equation}
where $\Delta_i$ is the signed value correction, the difference between
player $i$'s weighted average
realized value and the value of the average strategy profile. It may be
negative.
\end{theorem}

\begin{proof}
Fix $i$ and add and subtract $\sum_t \rho_t u_i(\sigma^t)$: the comparator
term is then handled exactly by bilinearity (valid because $n=2$), the
value term satisfies $\sum_t \rho_t u_i(\sigma^t) = \ip{\dbar}{u_i}$, and
taking $\max_i$ yields \eqref{eq:tight-identity}. The algebra is in
Appendix~\ref{app:proofs}.
\end{proof}

\begin{corollary}[Online certificate, simultaneous updates, $n=2$]
\label{cor:online-cert}
For $n=2$ define $C^{\mathrm{tight}} \coloneqq \max_i \big[ \Rhat_i + \Delta_i \big]$.
When the raw signed regret accumulators, averages, and terminal accumulators
refer to the same profile sequence, $C^{\mathrm{tight}} \ge \gap(\sbar)$.
Every ingredient is online: $\Rhat_i$ from the raw weighted accumulator
of Lemma~\ref{lem:weighted-decomp}, an $O(\sum_i |\cI_i| A)$ scan;
$\ip{\dbar}{u_i}$ from the running realized-value sum; and
$u_i(\sbar)=\ip{\mubar}{u_i}$ from the average realization plans, $O(|\cZ|)$
per checkpoint and no best response. The only slack relative to
\eqref{eq:tight-identity} is $\Rhat_i-\Rbar_i$.
\end{corollary}

Lemma~\ref{lem:exchange} gives $\Delta_i = \ip{\dbar - \mubar}{u_i} \le
M \lVert \dbar - \mubar \rVert_1 \le M\sqrt{2\eps}$.
So $C^{\mathrm{tight}}$ is no larger than the distribution-exchange
certificate termwise whenever both sides use the same regret proxy $\Rhat_i$,
$C^{\mathrm{tight}} = \max_i[\Rhat_i + \Delta_i] \le \max_i[\Rhat_i + M\sqrt{2\eps}]$.
For simultaneous two-player runs this is the deployable certificate. The
alternating updates used by the strongest solvers need one further comparator
term, supplied by Proposition~\ref{prop:sequence} below, a signed linear
functional over the terminal weights that measures the mismatch between the
profile a regret was evaluated on and the profile the averages record.

\textbf{The coalition term for three or more players.}
\label{sec:coalition}

For $n \ge 3$ the comparator step of Theorem~\ref{thm:tight} fails.
Indeed $\prod_{j \ne i} \xbar_j \ne \sum_t \rho_t \prod_{j \ne i} x_j^t$ in
general, because averaging destroys independence \emph{within the opposing
coalition}. The deviation is again directly measurable. Define the
\emph{coalition imbalance}
\begin{equation}
  D_{-i} \;\coloneqq\; \sum_{z \in \cZ} p_c(z)\,
  \Big| \bar{c}_{-i}(z) \;-\; \prod_{j \ne i} \xbar_j(z) \Big|,
\end{equation}
where $\bar{c}_{-i}(z) \coloneqq \sum_t \rho_t \prod_{j \ne i} x_j^t(z)$
is maintained online by accumulating the coalition products (cost
$O(n|\cZ|)$ per iteration).

\begin{proposition}[Coalition certificate, general $n$]
\label{prop:coalition}
For any number of players and nonnegative weights,
\begin{equation}
  \gap(\sbar) \;\le\; \max_i
  \Big[\, \Rbar_i + \Delta_i + M \, D_{-i} \,\Big],
\end{equation}
and for $n = 2$, $D_{-i} \equiv 0$, recovering the bound of
Theorem~\ref{thm:tight}.
\end{proposition}

The deployable coalition certificate is
\begin{equation}
 C^{\mathrm{coal}} \coloneqq \max_i
 [\Rhat_i+\Delta_i+M D_{-i}],
\end{equation}
sound when the raw signed regret, average-strategy, empirical-terminal,
and coalition accumulators all describe the same profile sequence. The
$\Rbar_i$ form in the proposition is a theoretical bound, and the
certificate uses the measured signal.

The correlation level $\eps$ is structurally blind to $D_{-i}$, which is not a
function of $\eps$. The certificate therefore measures
$D_{-i}$ directly. Remark~\ref{rem:eps-blind} and
Appendix~\ref{app:proofs} give the details.

\begin{remark}[$\eps$ is blind to the coalition term]
\label{rem:eps-blind}
A bound on $D_{-i}$ that vanishes with $\eps = \KL(\dbar \Vert \mubar)$
cannot follow from $\eps$ alone: at terminals $z$ with $\xbar_i(z)=0$,
both $\dbar(z)$ and $\mubar(z)$ vanish whatever the opposing coalition
does there, so two configurations differing only in the coalition's
correlation at such $z$ share the same $\eps$ and differ in $D_{-i}$, while
a deviation $x'$ can still reach $z$. The coalition term has to be measured
directly, since the correlation level cannot bound it.
Appendix~\ref{app:eps-blind-game} gives a three-player game with $\eps=0$,
all regrets $0$ and $\gap(\sbar)=\tfrac12$.
\end{remark}

The factor $M$ in Proposition~\ref{prop:coalition} is the only remaining
relaxation of the comparator term, and it can be removed exactly.

\begin{proposition}[Exact coalition comparator]
\label{prop:coalition-lbr}
Define the signed imbalance weight
$w_i(z) \coloneqq p_c(z)\, u_i(z)
 \big( \prod_{j \ne i} \xbar_j(z) - \bar{c}_{-i}(z) \big)$
and $L_{-i} \coloneqq \max_{x' \in \cX_i} \ip{x'}{w_i}$. Then for any
number of players and nonnegative weights,
\begin{equation}
\begin{aligned}
  \gap(\sbar) &\;\le\; \max_i
  \Big[\, \Rbar_i + \Delta_i + L_{-i} \,\Big], \\
  L_{-i} &\;\le\; M \, D_{-i}.
\end{aligned}
\end{equation}
and $L_{-i}$ can be negative.
\end{proposition}

$L_{-i}$ is one backward pass per player over the logged imbalance weight
$w_i$, a signed linear functional maximized over $\cX_i$. It also
attributes the zero-best-response certificate's slack at the certificate-maximizing
player $i$, $C^{\mathrm{coal}} - \gap(\sbar) = (\Rhat_i - \Rbar_i) + (M D_{-i} - L_{-i})
+ \xi_i$, where the residual $\xi_i$ is the slack of splitting the maximum
over $x'$ in the proof, less $\gap(\sbar) - \mathrm{gap}_i$. On three-player
games the comparator relaxation
$M D_{-i} - L_{-i}$ accounts for $96$--$99\%$ of that slack
(Section~\ref{sec:exp-3p}).

\textbf{Any update order.}
\label{sec:update-order}
The statements so far measure player $i$'s regret against the profile
sequence $\sigma^t$ that the averages use. Alternating updates
\citep{burch2019revisiting}, which the strongest solvers rely on, evaluate
player $i$ after players $1,\dots,i-1$ have moved, so its regrets refer to
a partially updated opponent and the $\Rhat_i$ read off them need not
dominate $\Rbar_i$. A second regret pass per iteration against the
start-of-round profile is the usual remedy. The mismatch itself is one
more measurable comparator term, so no second pass is needed. Write
$y^{(i),t}_{-i}$ for the opponent
profile against which $ir_i^t$ was computed ($y^{(i),t}_j = x_j^{t+1}$ for
$j<i$ and $x_j^t$ for $j>i$ under alternation; $x^t_{-i}$ under
simultaneous updates) and accumulate, when that regret is computed,
\begin{equation}
  \dbar^{(i)}(z) \coloneqq \sum_t \rho_t\, p_c(z)\, x_i^t(z)
    \prod_{j \ne i} y^{(i),t}_j(z), \qquad
  \bar c^{(i)}_{-i}(z) \coloneqq \sum_t \rho_t \prod_{j \ne i} y^{(i),t}_j(z).
\end{equation}

\begin{proposition}[Certificate along the evaluation sequence]
\label{prop:sequence}
For any number of players, nonnegative weights, and evaluation sequences,
with $\Rhat_i$ the bound of Lemma~\ref{lem:weighted-decomp} applied to the
regrets $ir_i^t$ measured at $(x_i^t, y^{(i),t}_{-i})$,
\begin{equation}
\label{eq:seq-cert}
\begin{aligned}
  \gap(\sbar) &\;\le\; \max_i \big[\, \Rhat_i + \Delta^{(i)}_i
    + L^{(i)}_{-i} \,\big], \qquad
  \Delta^{(i)}_i \coloneqq \ip{\dbar^{(i)}}{u_i} - u_i(\sbar), \\
  w^{(i)}_i(z) &\coloneqq p_c(z)\, u_i(z)
    \Big( \prod_{j \ne i} \xbar_j(z) - \bar c^{(i)}_{-i}(z) \Big), \qquad
  L^{(i)}_{-i} \coloneqq \max_{x' \in \cX_i} \ip{x'}{w^{(i)}_i},
\end{aligned}
\end{equation}
and $L^{(i)}_{-i} \le \lVert w^{(i)}_i \rVert_1 \le M D^{(i)}_{-i}$ with
$D^{(i)}_{-i} \coloneqq \sum_z p_c(z)\,|\prod_{j\ne i}\xbar_j(z) -
\bar c^{(i)}_{-i}(z)|$.
\end{proposition}

The proof in Appendix~\ref{app:proofs} is that of
Proposition~\ref{prop:coalition-lbr} with $\bar c_{-i}$ replaced by
$\bar c^{(i)}_{-i}$. The comparator identity holds for any sequence. The
regret step is Lemma~\ref{lem:weighted-decomp} on the sequence player $i$
was evaluated on. The value term is $\ip{\dbar^{(i)}}{u_i}$ by
definition. Simultaneous updates recover Corollary~\ref{cor:online-cert}
($w_i \equiv 0$ for $n=2$) and
Propositions~\ref{prop:coalition}--\ref{prop:coalition-lbr}. Under
two-player alternation the first mover's weight vanishes and the second
mover's is a drift between consecutive iterates of the first,
\begin{equation}
\label{eq:drift}
  w^{(2)}_2 = p_c\, u_2 \sum_t \rho_t \big( x_1^t - x_1^{t+1} \big),
  \qquad
  \sum_t \rho_t \big( x_1^t - x_1^{t+1} \big)
  = \rho_1 x_1^1 + \sum_{t=2}^{T} (\rho_t - \rho_{t-1})\, x_1^t
    - \rho_T x_1^{T+1},
\end{equation}
so for nondecreasing weights the coefficients $\rho_1$ and
$\rho_t - \rho_{t-1}$ are nonnegative and sum to $\rho_T$, and
$\sum_t \rho_t (x_1^t - x_1^{t+1}) = \rho_T(\tilde x_1 - x_1^{T+1})$ with
$\tilde x_1 \in \cX_1$ a convex combination of the iterates. Hence for
every deviation $x'$,
$\ip{x'}{w^{(2)}_2} = \rho_T\big(u_2(x', \tilde x_1) - u_2(x', x_1^{T+1})\big)$,
so $|L^{(2)}_{-2}| \le 2M\rho_T$, which is $O((p+1)/T)$ for $w_t = t^p$,
the scale of one iteration's weight, against the $O(1/\sqrt{T})$
worst-case regret bound of CFR \citep{zinkevich2007regret}. The drift therefore decays faster than
that bound at long horizons. This is a statement about the bound, not about
the observed regret or gap. This drift is the
alternation correction that \citet[Theorem~1]{burch2019revisiting} write as a trailing
sum for zero-sum uniform averages and discard through an improved-strategy monotonicity
argument. We keep it as a measured signed term: no such monotonicity holds in
general-sum play, and at every finite $T$ the certificate reports it. Every
ingredient of \eqref{eq:seq-cert} is a byproduct of the update: the
opponent reach $p_c \prod_{j \ne i} y_j$ is formed to evaluate $ir_i^t$, so
$\dbar^{(i)}$ and $\bar c^{(i)}_{-i}$ cost two $O(|\cZ|)$ accumulations
per player per iteration and no tree pass. We call
\begin{equation}
\label{eq:C0-C1}
  C^{0} \coloneqq \max_i \big[\Rhat_i + \Delta^{(i)}_i
    + \lVert w^{(i)}_i \rVert_1\big], \qquad
  C^{1} \coloneqq \max_i \big[\Rhat_i + \Delta^{(i)}_i + L^{(i)}_{-i}\big]
\end{equation}
the zero-best-response certificate and its exact-comparator refinement.
$C^1$ costs one signed backward pass per player with $w^{(i)}_i \ne 0$ at
each checkpoint, a single pass under two-player alternation, and never a best
response against the average. The certificate is a stated rule whose terms
attribute its bound. The same backward pass, applied to the raw regret
table and $w^{(i)}_i$ together, returns the exact gap
(Proposition~\ref{prop:exact-from-logs},
Appendix~\ref{app:exact-from-logs}); $C^1$ keeps the terms apart.

\begin{corollary}[One-shot best-response verification]
\label{cor:one-shot}
Consider the protocol: run $K$ solvers; freeze one stored snapshot
$\hat{\sigma}$ using an online signal, without consulting its exact gap;
then compute one exact best response per player at that snapshot,
obtaining $\hat{g}=\gap(\hat{\sigma})$. The pair $(\hat{\sigma},\hat g)$ reports
its exact NashGap. The selection signal and
the verification are operationally separate, so the signal may be any online rule.
The cost is one gap evaluation of $n$ best-response passes, against one at
every checkpoint in trajectory-minimum reporting, and the result holds for
any certificate slack, update scheme and player count.
\end{corollary}

Three layers recur throughout. Layer~I is the exact identity of
Theorem~\ref{thm:tight}: it contains $\Rbar_i$ and is an analysis statement.
Layer~II is the online certificate of
Proposition~\ref{prop:sequence}, a theorem-backed upper bound for every
update scheme and player count. $C^0$ is read off the regret accumulators and
the terminal accumulators kept alongside the update, and $C^1$ adds one signed backward pass per player with
nonzero comparator weight at each checkpoint.
Layer~III is one-shot verification, $n$ best-response passes on the
selected snapshot alone. The start-of-round \emph{aligned} certificate,
which logs a second regret pass per iteration so that every accumulator
refers to $\sigma^t$, is also a Layer~II bound. Section~\ref{sec:exp-main}
finds it looser than $C^1$ on the same trajectories
(Table~\ref{tab:guarantees}).

\section{Certificate-Guided Portfolio Selection}
\label{sec:portfolio}
Algorithm~\ref{alg:portfolio} turns Proposition~\ref{prop:sequence} into a
selection rule. Each solver keeps, next to its clipped optimizer state, the raw
accumulators behind the certificate: the weighted realization plans $X$, the
joint reach $D_i$ and opponent-coalition reach $H_i$ of the profile each
player was evaluated on, and the unclipped signed regret $Q_i$ of those
evaluations. These are the unnormalized counterparts, scaled by $W_T$,
of $\xbar_i$, $\dbar^{(i)}$ and $\bar c^{(i)}_{-i}$ of
Section~\ref{sec:certificates}. At a checkpoint they define
both the stored profile and its score.

\begin{algorithm}[t]
\caption{Certificate-guided portfolio selection and verification}
\label{alg:portfolio}
\begin{algorithmic}[1]
\REQUIRE game $\Gamma$, solvers $\{A_k\}_{k=1}^K$, horizon $T$,
 checkpoints $\mathcal T$, score $C\in\{C^0,C^1\}$
\FOR{each solver $A_k$ independently on game $\Gamma$; $W\leftarrow0$ (the running weight sum $W_t=\sum_{s\le t}w_s$); raw $Q_i,X_i,D_i,H_i\leftarrow0$; $s_k^\star\leftarrow+\infty$}
  \FOR{$t=1,\ldots,T$; $W{+}{=}w_t$; $X_i{+}{=}w_tx_i^t$ for all $i$}
    \FOR{each player $i$ in the solver's update order, with $y^{(i),t}$ the profile player $i$ is evaluated on}
      \STATE $H_i{+}{=}w_t\prod_{j\ne i}y^{(i),t}_j$;\;
        $D_i{+}{=}w_t\,p_c\,x_i^t\prod_{j\ne i}y^{(i),t}_j$;\;
        $Q_i{+}{=}w_t\,ir_i^t$ \COMMENT{same pass as the update}
      \STATE perform $A_k$'s update of player $i$; alternating solvers apply it
        before the next player is evaluated, simultaneous solvers after the loop
        (so $y^{(i),t}=x^t$ for all $i$)
    \ENDFOR
    \IF{$t\in\mathcal T$}
      \STATE $s_k(t)\leftarrow C$ from $(Q,X,D,H)/W$ via
        \eqref{eq:seq-cert}--\eqref{eq:C0-C1}; store profile $S_k(t)$ \COMMENT{Layer II}
      \STATE \textbf{if} $s_k(t)<s_k^\star$ \textbf{then}
        $(s_k^\star,S_k^\star)\leftarrow(s_k(t),S_k(t))$
    \ENDIF
  \ENDFOR
\ENDFOR
\STATE $\hat k\leftarrow\arg\min_k s_k^\star$;
 $(\hat\sigma,\hat s)\leftarrow(S_{\hat k}^\star,s_{\hat k}^\star)$;
 \textbf{optional:} $\hat g\leftarrow\gap(\hat\sigma)$ with $n$ BR passes
  \COMMENT{Layer III}
\RETURN $(\hat\sigma,\hat s)$, or $(\hat\sigma,\hat g)$ if verified
\end{algorithmic}
\end{algorithm}

\section{Experiments}
\label{sec:experiments}

\textbf{Setup.}
We evaluate on the two extensive-form benchmarks of
\citet{sychrovsky2026approximating}. \emph{Biased Leduc}: two-player Leduc
poker in which on a tie each player recovers only a fraction
$\beta \sim U(0, \tfrac12)$ of their wager, making the game general-sum,
with $64$ instances from paired seeds, $T = 2^{18}$, and our own dense grid
of $128$ checkpoints (the reference reports about $19$). \emph{Three-player
Leduc}: the fixed game of their Table~6, $T = 2^{16}$. The portfolio has
$K = 16$ training-free regret-matching solvers, the seven original RM$^+$,
PRM$^+$ and DPRM$^+$ solvers under alternating updates with polynomial averaging
\citep{tammelin2014solving,farina2021faster,burch2019revisiting}, and nine
solvers added in a second stage, one from each of DCFR \citep{brown2019solving},
DCFR$^+$ and PDCFR$^+$ \citep{xu2024minimizing}, the hyperparameter-schedule
variants HS-PCFR$^+$ and HS-DCFR \citep{zhang2026faster}, Smooth, Stable and
extragradient (ExRM$^+$) regret matching$^+$ \citep{farina2023regret}, and
IR-PRM$^+$ \citep{zhang2025scale}. Each added solver is the configuration of its
family that solved most pilot instances (seeds $0$--$9$) under a rule frozen
before the pilot ran; the in-sample numbers below therefore include that
selection step, the held-out numbers do not. For each run we call the
smallest exact NashGap over its checkpoints the \emph{trajectory minimum},
and an \emph{oracle} row reports the fraction of instances whose trajectory
minimum clears the threshold. The solver pool, the definitions of the selection
and certified rows, and the recomputation of every reported number from the
stored logs are in Appendix~\ref{app:impl}.

\begin{table}[tp]
\centering
\small
\setlength{\tabcolsep}{3pt}
\caption{Biased Leduc, solved fractions at four thresholds on the $64$
paired instances. \emph{Reported} rows are figures of
\citet{sychrovsky2026approximating}; the rest are our runs. Single-solver rows
are trajectory-minimum oracles. Selection rows give exact gaps, and certified
rows count certificates below the threshold. Control rows are outside both
portfolios; the full names of the added solvers are in
Appendix~\ref{app:impl}.}
\label{tab:main}
\begin{tabular*}{\textwidth}{@{\extracolsep{\fill}}lcccc@{}}
\toprule
 & \multicolumn{4}{c}{fraction of games with NashGap $\le$} \\
\cmidrule(lr){2-5}
Method & $10^{-2}$ & $10^{-3}$ & $10^{-4}$ & $10^{-5}$ \\
\midrule
\multicolumn{5}{@{}l}{\emph{Reported}} \\
CFR$^+$ & $1$ & $0$ & $0$ & $0$ \\
PCFR$^+$ & $1$ & $0.87$ & $0.48$ & $0.16$ \\
NPCFR$^+$ (oracle) & $1$ & $0.98$ & $0.96$ & $\mathbf{0.96}$ \\
\midrule
\multicolumn{5}{@{}l}{\emph{Controls (outside the pool)}} \\
RM$^+$, uniform & $1.00$ & $0.33$ & $0.00$ & $0.00$ \\
PRM$^+$, uniform & $1.00$ & $1.00$ & $0.77$ & $0.00$ \\
PRM$^+$, $p{=}2$, simult. & $1.00$ & $0.73$ & $0.64$ & $0.41$ \\
\midrule
\multicolumn{5}{@{}l}{\emph{Original seven solvers}} \\
RM$^+$, $p{=}1$ & $1.00$ & $0.50$ & $0.00$ & $0.00$ \\
RM$^+$, $p{=}2$ & $1.00$ & $0.52$ & $0.00$ & $0.00$ \\
PRM$^+$, $p{=}1$ & $1.00$ & $1.00$ & $0.88$ & $0.67$ \\
PRM$^+$, $p{=}2$ & $1.00$ & $1.00$ & $0.89$ & $0.70$ \\
DPRM$^+$, $p{=}1$, $\lambda{=}0.3$ & $1.00$ & $1.00$ & $0.88$ & $0.73$ \\
DPRM$^+$, $p{=}2$, $\lambda{=}0.1$ & $1.00$ & $1.00$ & $0.88$ & $0.67$ \\
DPRM$^+$, $p{=}2$, $\lambda{=}0.3$ & $1.00$ & $1.00$ & $0.91$ & $0.75$ \\
\midrule
\multicolumn{5}{@{}l}{\emph{Added nine solvers}} \\
DCFR, $p{=}2$ & $1.00$ & $0.62$ & $0.02$ & $0.00$ \\
DCFR$^+$, $p{=}4$ & $1.00$ & $0.42$ & $0.00$ & $0.00$ \\
PDCFR$^+$, $p{=}5$ & $1.00$ & $1.00$ & $0.75$ & $0.66$ \\
HS-PCFR$^+$ & $1.00$ & $1.00$ & $0.80$ & $0.72$ \\
HS-DCFR & $1.00$ & $0.52$ & $0.00$ & $0.00$ \\
Smooth PRM$^+$, $p{=}2$ & $1.00$ & $1.00$ & $0.88$ & $0.73$ \\
Stable PRM$^+$, $p{=}2$ & $1.00$ & $1.00$ & $0.89$ & $0.80$ \\
ExRM$^+$, $p{=}2$, simult. & $1.00$ & $0.94$ & $0.53$ & $0.27$ \\
IR-PRM$^+$, $p{=}2$ & $1.00$ & $0.81$ & $0.28$ & $0.19$ \\
\midrule
Portfolio of $7$ solvers: selection by $C^1$ & $1.00$ & $1.00$ & $0.91$ & $0.84$ \\
\quad selection by aligned certificate & $1.00$ & $1.00$ & $0.91$ & $0.80$ \\
\quad certified, $C^1 \le$ threshold & $1.00$ & $1.00$ & $0.89$ & $0.75$ \\
\quad certified, $C^0 \le$ threshold & $1.00$ & $1.00$ & $0.88$ & $0.73$ \\
\quad certified, aligned $\le$ threshold & $1.00$ & $1.00$ & $0.88$ & $0.69$ \\
\midrule
Portfolio of $16$ solvers: selection by $C^1$ & $1.00$ & $1.00$ & $0.97$ & $\mathbf{0.89}$ \\
\quad selection by aligned certificate & $1.00$ & $1.00$ & $0.97$ & $0.89$ \\
\quad certified, $C^1 \le$ threshold & $1.00$ & $1.00$ & $0.97$ & $\mathbf{0.89}$ \\
\quad certified, $C^0 \le$ threshold & $1.00$ & $1.00$ & $0.97$ & $0.89$ \\
\quad certified, aligned $\le$ threshold & $1.00$ & $1.00$ & $0.95$ & $0.86$ \\
\bottomrule
\end{tabular*}
\end{table}

\label{sec:exp-main}

Table~\ref{tab:main} holds the central result. The reported CFR$^+$ anchor
solves nothing at $10^{-3}$ and below, and the controls trail the pool. At
$10^{-5}$ the controls solve $0.00$, $0.00$ and $0.41$ against $0.89$ for the
certificate selection, and on
three-player Leduc the uniform RM$^+$ anchor sits about $2.25\times$ below
the reported alternating CFR$^+$ (Table~\ref{tab:leduc3}). Uniform
averaging caps RM$^+$ and PRM$^+$ at $10^{-5}$; predictiveness lifts RM$^+$ to an
oracle fraction of $0.70$, and the strongest single solver of the extension,
Stable PRM$^+$, reaches $0.80$. Regret discounting alone does not help. DCFR,
DCFR$^+$ and HS-DCFR solve $0.00$ at $10^{-5}$ (Table~\ref{tab:perarm}).

\begin{figure}[t]
\centering
\includegraphics[width=0.46\figbase]{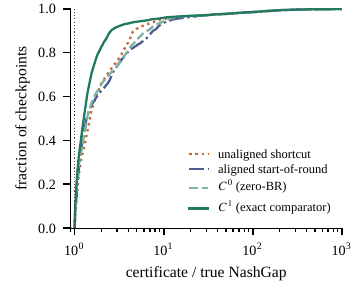}
\caption{Tightness over the $16\times64\times128$ biased-Leduc
checkpoints whose exact gap is above floating-point zero (excludes
$3{,}223$ checkpoints, all but $31$ on HS-PCFR$^+$): an ECDF of certificate
over true NashGap for $C^1$ and $C^0$ (Proposition~\ref{prop:sequence}), the
aligned start-of-round certificate, and the unaligned shortcut of
Corollary~\ref{cor:online-cert} under alternation. Sound certificates lie at
ratio $\ge 1$. Iteration traces are in Figures~\ref{fig:hard-instances}
and~\ref{fig:tightness-curves}.}
\label{fig:tightness}
\end{figure}

Selection by $C^1$ solves $57/64$ at $10^{-5}$, exactly the
trajectory-minimum union of the sixteen solvers over $16\times128$
checkpoints.

Within fourteen of the solvers the certificate argmin lands on an iterate at
or below the threshold whenever the solver's best checkpoint does; HS-PCFR$^+$
and ExRM$^+$ each miss their own best checkpoint on one instance
(Table~\ref{tab:perarm}), and another solver covers the instance. The same grid supports
an exact-gap selector, which attains that union at $128$ gap evaluations per
run (Table~\ref{tab:budget}). Per instance, the selected gap exceeds the exact
trajectory minimum by a median factor of $1.01$ ($1.36$ at the $90$th
percentile) on the $36$ in-sample instances whose minimum is above
$10^{-12}$; on the remaining $28$ both sit at or below $10^{-12}$
(Table~\ref{tab:selection-loss}). Selection by $C^0$, and by
$\max_i[\Rhat_i+\Delta^{(i)}_i]$ alone, gives the same $57$
(Appendix~\ref{app:pool-ablation}).

The certified rows report what a solver can state without a best response
against the reported profile. On this set the sixteen-solver pool certifies
every instance it solves, $57/64$ under both $C^1$ and $C^0$. On the
seven-solver pool certified counts trail selection by six or seven instances,
$48/64$ for $C^1$ and $47/64$ for $C^0$ against $54$ solved, with
certificates between $10^{-5}$ and $1.3\times10^{-5}$. The added solvers close
that margin because the selector now lands on iterates whose certificate is
well below the threshold.

The selected solver is HS-PCFR$^+$ on $28$ instances (PDCFR$^+$ $11$, ExRM$^+$
$6$, Smooth and Stable PRM$^+$ $5$ each, IR-PRM$^+$ $4$, an original solver
$5$), and on the $28$ HS-PCFR$^+$ selections the exact gap is at
floating-point zero, below $10^{-12}$. The reported $0.96$ is a
trajectory-minimum oracle. Our selection reaches $0.89$ on the instances
above, $0.92$ on $64$ held-out instances, and $0.89$ again on a later
instance set that no earlier stage had used, all from the same distribution
(Tables~\ref{tab:heldout} and~\ref{tab:fresh}).

\textbf{Three-player Leduc.}
\label{sec:exp-3p}
On the single deterministic three-player benchmark twelve solvers were run,
the six original RM$^+$/PRM$^+$ solvers and the six added solvers whose two-player
oracle fraction or family membership qualified them (Appendix~\ref{app:impl}).
PDCFR$^+$ reaches true NashGap $1.17\times10^{-6}$ and HS-PCFR$^+$
$2.16\times10^{-6}$, both below the $5.50\times10^{-6}$ of quadratically
weighted alternating PRM$^+$ and three orders below the reported $10^{-3}$
NPCFR$^+$ comparator. With no aligned rerun, $C^1$ certifies
$2.83\times10^{-6}$ for HS-PCFR$^+$ at $T$ ($1.31\times$ the true gap) and
$7.65\times10^{-6}$ for PRM$^+$ ($1.39\times$); on the PRM$^+$ solver $C^0$
gives $1.59\times10^{-5}$ and the $M D_{-i}$ relaxation $6.42\times10^{-5}$.
At $T$ $C^1$ stays within $1.31$--$8.2\times$ of the true gap on all twelve
solvers, and all three certificates were above the exact gap at every one of
the $64$ checkpoints of every solver (Table~\ref{tab:leduc3}). The cross-solver
$C^1$ argmin selects HS-PCFR$^+$ at $T$, whose gap is $1.84\times$ the
pool's trajectory minimum. The certificate picks the most tightly certified
run. The smallest gap belongs to PDCFR$^+$, which carries the looser
certificate $3.26\times10^{-6}$. Certificate selection and the trajectory
minimum therefore no longer coincide on this game once the pool is extended,
and the selected certificate is still $354\times$ below the reported
comparator.

\begin{table}[t]
\centering
\small
\setlength{\tabcolsep}{3.5pt}
\caption{Three-player Leduc ($T=2^{16}$, deterministic single runs). The
min NashGap column is the trajectory minimum with the checkpoint $t$ at
which it is attained; $C^1$, $C^0$ and the $M D$ relaxation are the
certificates of Proposition~\ref{prop:sequence} at $T$, with ratios to the
true gap at $T$. All solvers use alternating updates. Reported rows are from
Table~6 of \citet{sychrovsky2026approximating}; \emph{n/a} marks quantities
that study does not publish.}
\label{tab:leduc3}
\begin{tabular}{@{}lccccccc@{}}
\toprule
Solver & min NashGap & at $t$ & $C^1$ at $T$ & ratio & $C^0$ & ratio & $M D$ form \\
\midrule
\multicolumn{8}{@{}l}{\emph{Added solvers}} \\
PDCFR$^+$, $p{=}5$ & $\mathbf{1.17 \times 10^{-6}}$ & $59392$ & $3.26\times10^{-6}$ & $2.41\times$ & $1.73\times10^{-5}$ & $12.8\times$ & $8.59\times10^{-5}$ \\
HS-PCFR$^+$ & $2.16\times10^{-6}$ & $T$ & $\mathbf{2.83\times10^{-6}}$ & $1.31\times$ & $2.83\times10^{-6}$ & $1.3\times$ & $2.85\times10^{-6}$ \\
Stable PRM$^+$, $p{=}2$ & $1.59\times10^{-5}$ & $T$ & $4.29\times10^{-5}$ & $2.70\times$ & $3.76\times10^{-4}$ & $23.7\times$ & $1.02\times10^{-3}$ \\
DCFR$^+$, $p{=}4$ & $7.68\times10^{-5}$ & $33792$ & $3.25\times10^{-4}$ & $3.75\times$ & $4.08\times10^{-3}$ & $47.1\times$ & $1.57\times10^{-2}$ \\
DCFR, $p{=}2$ & $1.11\times10^{-4}$ & $15360$ & $7.68\times10^{-4}$ & $3.34\times$ & $9.05\times10^{-3}$ & $39.4\times$ & $3.60\times10^{-2}$ \\
Smooth PRM$^+$, $p{=}2$ & $6.51\times10^{-4}$ & $53248$ & $3.64\times10^{-3}$ & $5.37\times$ & $2.47\times10^{-2}$ & $36.5\times$ & $1.04\times10^{-1}$ \\
\midrule
\multicolumn{8}{@{}l}{\emph{Original solvers}} \\
PRM$^+$, $p{=}2$ & $5.50 \times 10^{-6}$ & $T$ & $7.65\times10^{-6}$ & $1.39\times$ & $1.59\times10^{-5}$ & $2.9\times$ & $6.42\times10^{-5}$ \\
PRM$^+$, $p{=}1$ & $5.63\times10^{-6}$ & $T$ & $2.64\times10^{-5}$ & $4.7\times$ & $3.07\times10^{-4}$ & $55\times$ & $1.24\times10^{-3}$ \\
PRM$^+$, $p{=}0$ & $1.98\times10^{-4}$ & $T$ & $1.62\times10^{-3}$ & $8.2\times$ & $2.08\times10^{-2}$ & $105\times$ & $7.67\times10^{-2}$ \\
RM$^+$, $p{=}2$ & $1.06\times10^{-4}$ & $T$ & $3.39\times10^{-4}$ & $3.2\times$ & $3.23\times10^{-3}$ & $30\times$ & $1.10\times10^{-2}$ \\
RM$^+$, $p{=}1$ & $2.61\times10^{-4}$ & $T$ & $5.57\times10^{-4}$ & $2.1\times$ & $4.77\times10^{-3}$ & $18\times$ & $1.78\times10^{-2}$ \\
RM$^+$, $p{=}0$ & $1.78\times10^{-3}$ & $T$ & $3.83\times10^{-3}$ & $2.2\times$ & $3.87\times10^{-2}$ & $22\times$ & $1.53\times10^{-1}$ \\
\midrule
\multicolumn{8}{@{}l}{\emph{Reported methods (reference nomenclature)}} \\
CFR$^+$ & $4\times10^{-3}$ & n/a & n/a & n/a & n/a & n/a & n/a \\
PCFR$^+$ & $8\times10^{-3}$ & n/a & n/a & n/a & n/a & n/a & n/a \\
NPCFR$^+$ & $1\times10^{-3}$ & n/a & n/a & n/a & n/a & n/a & n/a \\
\bottomrule
\end{tabular}
\end{table}

\begin{figure}[t]
\centering
\includegraphics[width=0.76\figbase]{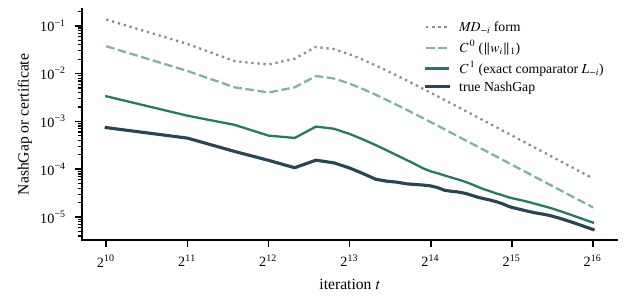}
\caption{Slack attribution on three-player Leduc (PRM$^+$, $p{=}2$,
alternating, log--log). The $M D_{-i}$ form of the certificate sits
$11.7$--$241\times$ above the true NashGap over training; the zero-BR
$C^0$ replaces $M D_{-i}$ by $\lVert w_i\rVert_1$ and tightens it to
$2.9$--$59\times$; the exact comparator $C^1$ (one backward pass per player
per checkpoint) reaches $1.39$--$5.2\times$.}
\label{fig:coalition}
\end{figure}

Figure~\ref{fig:coalition} attributes the slack on the quadratically
weighted PRM$^+$ solver, the run the six-solver pool of the first stage selected.
It reaches true NashGap $5.50\times10^{-6}$, and $C^1$ certifies
$7.65\times10^{-6}$ for it with no aligned rerun, $131\times$ below the best
reported oracle-protocol number. The uniform RM$^+$ anchor is
$1.78\times10^{-3}$ versus the reported CFR$^+$ value $4\times10^{-3}$. The
aligned start-of-round rerun of the $p{=}2$ solver reproduces the NashGap
trajectory to $8.3\times10^{-17}$; the start-of-round $M D$ form coincides
with its evaluation-sequence value at every checkpoint, while the underlying
regret terms differ by a median of $3.0\times10^{-4}$ of the gap (maximum
$0.30$). The slack is predominantly ($96$--$99\%$, median $0.98$) the
comparator relaxation, which $C^1$ removes. On HS-PCFR$^+$ the three certificates
nearly coincide ($C^1$ and $C^0$ both $2.83\times10^{-6}$, the $M D$ form
$2.85\times10^{-6}$). Its schedule puts
almost all averaging weight on the last iterates, so the coalition
imbalance that the comparator term prices is itself near zero.

\begin{table}[t]
\centering
\small
\setlength{\tabcolsep}{3pt}
\caption{Verification budgets (one gap evaluation equals $n$ best-response
passes). $C^1$ spends one signed backward pass per player whose comparator
weight is nonzero, the second mover alone under two-player alternation.
$C^1$'s pass maximizes a logged linear functional, not a best response
against the average; one-shot verification is exact for the selected
snapshot. The oracle
protocol uses about $19$ evaluations on the reference schedule and $128$
on our dense two-player grid. $C^1$ costs nothing extra when the comparator
weight vanishes, as it does for simultaneous updates in two-player games.}
\label{tab:verification-cost}
\begin{tabular}{@{}>{\raggedright\arraybackslash}p{0.26\columnwidth}
                    >{\centering\arraybackslash}p{0.44\columnwidth}
                    >{\raggedright\arraybackslash}p{0.22\columnwidth}@{}}
\toprule
Protocol & Exact-gap queries & Guarantee \\
\midrule
Selection by $C^0$ & $0$ & Layer II bound \\
Selection by $C^1$ & one signed backward pass per checkpoint under two-player alternation; $0$ for simultaneous updates & Layer II bound \\
One-shot verification & $1$ & exact $\gap(\hat\sigma)$ \\
Trajectory minimum & per checkpoint & evaluation protocol \\
\bottomrule
\end{tabular}
\end{table}

\begin{table}[tp]
\centering
\small
\setlength{\tabcolsep}{3pt}
\caption{Compute budget behind Table~\ref{tab:main} (biased Leduc). One gap
evaluation equals $n{=}2$ best-response passes, and the timed policy is
PRM$^+$, $p{=}2$. The isolated block times each term alone on one core
outside the training loop; the end-to-end and checkpoint-work blocks time
full runs with three repeats per policy in shuffled order
($\mathrm{mean}\,[\min,\max]$). \emph{n/a} marks a cell the measurement does
not cover; ``not reported'' marks a value the reference does not publish.
The \emph{bare solver, exact-gap selection} row is a later re-measurement,
with the uninstrumented baseline of its own batch at $81.9$ s.}
\label{tab:budget}
\begin{tabular}{@{}lccc@{}}
\toprule
 & Ours, $C^0$/$C^1$ & Ours, aligned & Reported \\
 & (Prop.~\ref{prop:sequence}) & (extra pass) & NPCFR$^+$ \\
\midrule
Training trajectories & $16$ solvers & $16$ solvers & $1$ \\
Instances & $64$ (paired) & $64$ (paired) & not stated \\
Iterations per run & $2^{18}$ & $2^{18}$ & $2^{18}$ \\
Tree passes per iteration & $n{=}2$ & $2n{=}4$ & not reported \\
Candidate snapshots & $16\times128{=}2048$ & $16\times128{=}2048$ & $\approx 19$ \\
Gap evals, selection & $0$ ($C^0$); one signed backward pass per ckpt.\ ($C^1$) & $0$ & $19$ \\
Gap evals, validation & $128$/run & $128$/run & n/a \\
\multicolumn{4}{@{}l}{\emph{Measured wall-clock, biased Leduc $\beta{=}0.25$, one core, isolated terms}} \\
One training iteration & $0.37$ ms & $0.51$ ms & not reported \\
One checkpoint ($C^0$; $C^1$) & $0.12$; $0.19$ ms & $0$ extra & not reported \\
One exact gap (validation) & $0.15$ ms & $0.15$ ms & not reported \\
\multicolumn{4}{@{}l}{\emph{End-to-end run, $T{=}2^{18}$, one core, $128$ checkpoints, three repeats ($\mathrm{mean}\,[\min,\max]$)}} \\
Uninstrumented solver & $91.9$ s $[83.4, 100.2]$ & n/a & not reported \\
Certificate state, no checkpoint work & $117.3$ s $[109.4, 128.6]$ & n/a & not reported \\
$C^0$ checkpoints & $116.2$ s $[113.0, 121.9]$ ($0.03$ s, $0.03\%$) & n/a & not reported \\
$C^1$ checkpoints & $118.6$ s $[108.8, 124.4]$ ($0.05$ s, $0.04\%$) & n/a & not reported \\
Exact-gap selection & $118.2$ s $[100.6, 130.9]$ ($0.03$ s, $0.02\%$) & n/a & not reported \\
Bare solver, exact-gap selection & $84.4$ s $[82.0, 87.9]$ ($0.02$ s, $0.03\%$) & n/a & not reported \\
\multicolumn{4}{@{}l}{\emph{Same policies, checkpoint every $8$ iterations ($32768$ total)}} \\
Checkpoint work & \multicolumn{3}{l}{$C^0$ $6.8$ s, $C^1$ $9.8$ s, exact gap $5.6$ s ($6.0$, $8.7$, $5.4\%$ of the run)} \\
Per checkpoint & \multicolumn{3}{l}{$0.21$ / $0.30$ / $0.17$ ms; peak RSS $\le 62$ MB} \\
\bottomrule
\end{tabular}
\end{table}

Per checkpoint, the certificate's own cost stays low. $C^0$ costs about
three quarters of an exact gap, and the $C^1$ backward pass is comparable to
one. Inside a full run these evaluations are negligible either way. The
$128$ validation evaluations cost $19$\,ms in total, about $0.02\%$ of the
run, and in situ each checkpoint costs $1.2$--$1.7\times$ its isolated
figure. What the certificate does cost is the accumulator bookkeeping,
$25.4$\,s, which is $21.7\%$ of an instrumented run and $27.6\%$ on top of the bare
solver. $C^0$ is no cheaper than an exact gap, and the four certificate-state
policies sit within run-to-run variation of each other.

The aligned variant spends the portfolio budget. It adds a second regret pass
on every iteration, $0.14$\,ms on a training iteration of $0.37$\,ms.
The $7\times64$ aligned grid of the original solvers is $24.4$ core-h and the
$9\times64$ grid of the added solvers $12.8$ core-h, $37.2$ core-h for the
$1024$ runs in total, measured end to end on the cluster including the
diagnostic evaluations. The added solvers run at $54$--$60$\,s per $2^{18}$
iterations, except ExRM$^+$ ($161$\,s, two projections per iteration) and
IR-PRM$^+$ ($134$\,s).

\begin{table}[t]
\centering
\small
\setlength{\tabcolsep}{3pt}
\caption{The three guarantee layers, with their best-response cost and
status. One gap evaluation equals $n$ best-response passes.}
\label{tab:guarantees}
\begin{tabular}{@{}>{\raggedright\arraybackslash}p{0.28\columnwidth}
                    >{\centering\arraybackslash}p{0.22\columnwidth}
                    >{\raggedright\arraybackslash}p{0.42\columnwidth}@{}}
\toprule
Layer & BR cost & Status \\
\midrule
I. Exact identity (2p) & exact $\Rbar$ & equality; not a deployment signal \\
II. $C^0$, zero-BR certificate & $0$ & theorem-backed upper bound, any update order and $n$ \\
II. $C^1$, exact comparator & $1$ signed pass per player with $w_i\ne0$, per ckpt. & theorem-backed upper bound; tighter \\
II. Aligned start-of-round bound & $1$ extra regret pass per player per iteration & theorem-backed upper bound; looser (ablation) \\
III. One-shot verification & $1$ gap & exact selected gap only \\
Trajectory minimum & per ckpt. & post-hoc evaluation protocol \\
\bottomrule
\end{tabular}
\end{table}

Alternating updates record the $i$-th updater's regret against a partially
updated opponent, whereas the averages use the start-of-round profile.
Proposition~\ref{prop:sequence} absorbs the mismatch into the comparator
weight $w^{(i)}_i$, so $C^0$ and $C^1$ are sound without an extra regret
pass per iteration. On the two-player grid they were above the exact gap
at every one of
the $131{,}072$ checkpoints, as the theorem requires. For comparison we also
paid the aligned pass across the full portfolio. Start-of-round regret
logging on all $16\times64$ runs (a rerun on identical trajectories for the
seven original solvers, the same run for the nine added ones) was likewise sound
at every checkpoint but certified fewer instances (Table~\ref{tab:main}).
The three-player runs report $C^1$ directly (Table~\ref{tab:leduc3}).

\textbf{Held-out instances.}
\label{app:heldout}
The main table's biased-Leduc column uses seeds $0$--$63$, the instances on
which the solver pool was chosen. To test whether the selection rule transfers,
we drew $64$ further instances (seeds $64$--$127$) from the same distribution
$\beta \sim U(0,\tfrac12)$ and evaluated the pool on them in two
pre-registered stages, each frozen before its held-out runs started, with the
Stage~1 document released alongside the code. Stage~1 ran the seven original
solvers with
$T = 2^{18}$ and $128$ checkpoints on the same engine and interpreter, and
fixed the three selectors ($C^1$ selection, regret-only selection and
same-grid exact-gap selection), the threshold $10^{-5}$ and a migration band
from $0.74$ to $0.91$. That band is the rounded Wilson $95\%$ interval
\citep{wilson1927probable} for $54/64$. Its normal-approximation
construction would give $0.75$ to $0.93$. Stage~2, after the pool was extended
on seeds $0$--$63$, fixed the nine added solvers, the same selectors and
threshold, and the band from $0.81$ to $0.97$. That band is the normal
approximation around the sixteen-solver in-sample estimate $57/64$, rounded
outward. The added solvers were then run on the held-out seeds without the
aligned pass, as in Stage~1.

\begin{table}[t]
\centering
\small
\setlength{\tabcolsep}{4pt}
\caption{Solved counts at $10^{-5}$ on the in-sample and held-out
instance sets for the seven-solver pool of Stage~1 and the sixteen-solver pool of
Stage~2, all from the stored checkpoint histories; brackets on the $C^1$ rows
are the frozen migration bands (in-sample) and $95\%$ Wilson intervals
(held-out). The $C^0$ and value-corrected rows are diagnostic ablations;
the frozen selectors are $C^1$, regret-only and exact-gap.}
\label{tab:heldout}
\begin{tabular}{@{}lcccc@{}}
\toprule
 & \multicolumn{2}{c}{$7$ solvers (Stage~1)} & \multicolumn{2}{c}{$16$ solvers (Stage~2)} \\
\cmidrule(lr){2-3}\cmidrule(lr){4-5}
Selector & seeds $0$--$63$ & seeds $64$--$127$ & seeds $0$--$63$ & seeds $64$--$127$ \\
\midrule
$C^1$ selection & $54/64$ \; $[0.74, 0.91]$ & $56/64$ \; $[0.77, 0.94]$ & $57/64$ \; $[0.81, 0.97]$ & $59/64$ \; $[0.83, 0.97]$ \\
$C^0$ selection & $54/64$ & $56/64$ & $57/64$ & $59/64$ \\
Value-corrected $\max_i[\Rhat_i+\Delta^{(i)}_i]$ & $54/64$ & $56/64$ & $57/64$ & $59/64$ \\
Exact-gap selection, same grid & $54/64$ & $56/64$ & $57/64$ & $59/64$ \\
Regret-only selection & $53/64$ & $54/64$ & $49/64$ & $56/64$ \\
Certified $C^1$ & $48/64$ & $55/64$ & $57/64$ & $59/64$ \\
Certified $C^0$ & $47/64$ & $54/64$ & $57/64$ & $58/64$ \\
Final average & $50/64$ & $54/64$ & $56/64$ & $58/64$ \\
Uniform random checkpoint & $0.41$ & $0.42$ & $0.35$ & $0.35$ \\
\bottomrule
\end{tabular}
\end{table}

Both stages land inside their bands. On Stage~1, $C^1$ selection solves
$56/64 = 0.875$ held-out instances, $2/64$ above the in-sample count.
Exact-gap selection on the same grid solves the same $56$, and the two agree
on the solved-or-not decision on $64/64$ instances, as they did in-sample.
Regret-only selection solves $54/64$, a subset of the $56$. Seeds $116$ and
$122$ are the two discordant instances, and the exact McNemar test
\citep{mcnemar1947note} gives $p = 0.5$, so the paired counts do not
establish a population-level advantage. The two instances are recovered by
the value correction $\Delta^{(i)}_i$ alone and equally by $C^0$.

On Stage~2, the sixteen-solver pool solves $59/64 = 0.922$ held-out
instances, again equal
to the exact-gap selection on the same grid and agreeing with it on $64/64$
instances. The extension adds $3$ solved instances in-sample and $3$
held-out and loses none in either set ($n_+ = 3$, $n_- = 0$, McNemar
$p = 0.25$ for each), so the paired counts show the direction and not yet the
size of the gain. Regret-only selection falls to $49/64$ in-sample and
$56/64$ held-out on the larger pool ($n_+ = 8$, $n_- = 0$, $p = 0.008$
in-sample; $3$, $0$, $p = 0.25$ held-out). With sixteen candidates the regret
term alone picks loosely certified runs more often, and the value and
comparator terms of $C^1$ keep the selection on the union. The added
solvers take $59$ of the $64$ in-sample selections and $56$ of the $64$ held-out
ones, against a pre-registered floor of $55$. The certified counts equal the
selection counts on both sets for $C^1$ ($57/64$ and $59/64$) and fall one
short held-out for $C^0$ ($58/64$). The unsolved seeds are $5, 8, 24, 29, 50,
53, 61$ in-sample and $76, 77, 85, 124, 126$ held-out, subsets of the Stage~1
lists ($5, 8, 24, 29, 34, 35, 50, 53, 60, 61$ and $65, 70, 76, 77, 85, 90,
124, 126$).

Stage~3 drew a third instance set, seeds $128$--$191$, from the same
distribution. The full sixteen-solver protocol was rerun on it under a
third pre-registration, frozen before those runs started, with the same
selectors, threshold and band. $C^1$ selection solves $57/64 = 0.891$,
inside the band and identical to the in-sample count. Exact-gap selection
on the same grid agrees on $64/64$ decisions, the added solvers take $59$
of the $64$ picks (floor $55$), and regret-only selection solves $53/64$
(Table~\ref{tab:fresh}). The paired gain over the seven-solver pool is
$n_+ = 2$, $n_- = 0$ (McNemar $p = 0.5$), and both certificates again
cover every solve. The unsolved fresh seeds are $144, 146, 150, 163, 170,
181, 184$, a subset of the seven-solver pool's $129, 142, 144, 146, 150,
163, 170, 181, 184$. Across the three instance sets the sixteen-solver
fraction is $0.891$, $0.922$ and $0.891$; the fresh point estimate equals
the in-sample one. It is $0.031$ below held-out, and the interval for that
difference, $-0.14$ to $+0.08$, contains zero.

\begin{table}[t]
\centering
\small
\setlength{\tabcolsep}{4pt}
\caption{Solved counts at $10^{-5}$ on the fresh instance set (seeds
$128$--$191$, Stage~3) for the seven-solver pool and the sixteen-solver
pool, from the stored checkpoint histories; the bracket on the $C^1$ row is
the frozen band and the row definitions are those of
Table~\ref{tab:heldout}.}
\label{tab:fresh}
\begin{tabular}{@{}lcc@{}}
\toprule
Selector & $7$ solvers & $16$ solvers \\
\midrule
$C^1$ selection & $55/64$ & $57/64$ \; $[0.81, 0.97]$ \\
$C^0$ selection & $54/64$ & $57/64$ \\
Value-corrected $\max_i[\Rhat_i+\Delta^{(i)}_i]$ & $55/64$ & $57/64$ \\
Exact-gap selection, same grid & $55/64$ & $57/64$ \\
Regret-only selection & $55/64$ & $53/64$ \\
Certified $C^1$ & $54/64$ & $57/64$ \\
Certified $C^0$ & $54/64$ & $57/64$ \\
Final average & $54/64$ & $57/64$ \\
Uniform random checkpoint & $0.43$ & $0.36$ \\
\bottomrule
\end{tabular}
\end{table}

\begin{table}[t]
\centering
\small
\caption{Selected gap relative to the exact trajectory minimum, per instance
(sixteen-solver pool; in-sample seeds $0$--$63$, held-out $64$--$127$, fresh
$128$--$191$). The ratio is reported over the instances whose minimum is above
$10^{-12}$; ``excluded'' counts the remaining instances, on all of which the
selected gap is also at or below $10^{-12}$.}
\label{tab:selection-loss}
\begin{tabular}{@{}lrrrr@{}}
\toprule
Instance set & Median & $90$th pct. & Max & Excluded \\
\midrule
In-sample & $1.01$ & $1.36$ & $42.45$ & $28$ \\
Held-out & $1.01$ & $2.67$ & $33.11$ & $26$ \\
Fresh & $1.00$ & $1.74$ & $17.07$ & $25$ \\
\bottomrule
\end{tabular}
\end{table}

\textbf{Where the certificate is loose.}
\label{app:detailed-results}
\label{app:hard-instances}

Figure~\ref{fig:hard-instances} characterizes where the training-free
portfolio stops. The two simultaneous solvers are covered by
Theorem~\ref{thm:tight} with no comparator term. The PRM$^+$ control
certifies $0.31$ at $10^{-5}$ with median tightness $1.3\times$, and ExRM$^+$
certifies $0.23$ at median tightness $1.24\times$ at its selected iterate.
The sixteen-solver portfolio's $7$ in-sample failures lean towards low $\beta$
(median $0.17$ against $0.28$ over all instances). Values $\beta>0.3$
cover $45\%$ of the benchmark and $29\%$ of the failures. On the
failures $C^1$ is within $1.10\times$ of the true gap at the median and
$1.22\times$ at worst, and no solver's trajectory minimum is below the
threshold on them either
(Table~\ref{tab:union-per-threshold}). One of the $7$ failures is borderline,
with gap $1.0\times10^{-5}$. The remaining six lie at
$1.7\times10^{-5}$--$1.7\times10^{-4}$.

The five held-out failures do not share the low-$\beta$ lean
(Table~\ref{tab:heldout}). Their median $\beta$ is $0.39$ against $0.32$ over
all held-out instances, with selected gaps from $4.7\times10^{-5}$ to
$1.8\times10^{-4}$ and $C^1$ within $1.35\times$ of each. The seven-solver pool
of the first stage failed on $10$ instances, at a median $\beta$ of $0.16$,
with $C^1$ within $1.15\times$ at the median and $2.2\times$ at worst. The
extension recovered seeds $34$, $35$ and $60$. The reference study's
unpublished, unpaired instance draws do not show whether NPCFR$^+$ improves
on the same instances.

\begin{figure}[t]
\centering
\includegraphics[width=0.9\figbase]{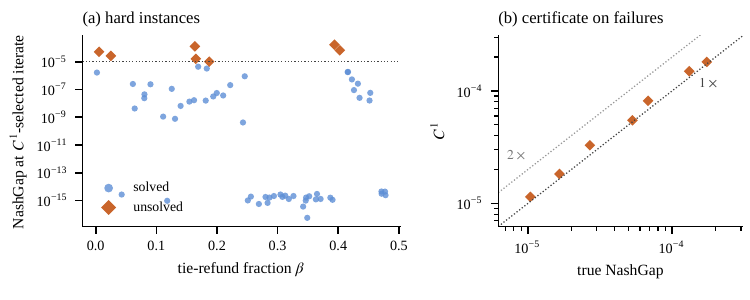}
\caption{Failure analysis (biased Leduc, sixteen-solver pool, threshold
$10^{-5}$). (a) NashGap at the $C^1$-selected iterate versus tie-refund
fraction $\beta$; the dotted line is the threshold, and the cluster at
$10^{-15}$ is floating-point zero, where HS-PCFR$^+$ is selected. (b) On
the seven unsolved instances, $C^1$ against the true gap, with the $1\times$
and $2\times$ lines.}
\label{fig:hard-instances}
\end{figure}

\begin{figure}[t]
\centering
\includegraphics[width=0.84\figbase]{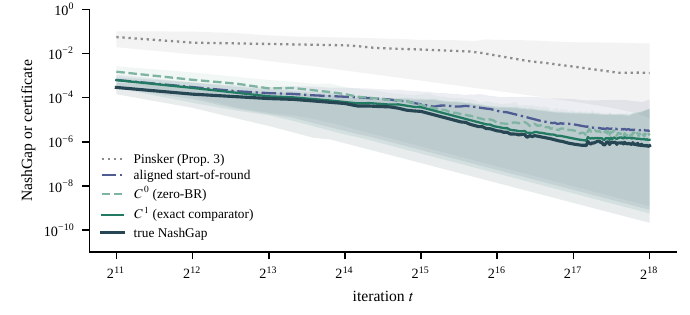}
\caption{Certificates over training (biased Leduc, PRM$^+$ $p{=}2$, median
and interquartile band over $64$ instances, log--log): the true NashGap,
$C^1$, $C^0$, the aligned start-of-round certificate and the Pinsker bound
of Proposition~\ref{prop:kl-cert}.}
\label{fig:tightness-curves}
\end{figure}

\textbf{What the comparator term does.}
\label{sec:exp-alignment}
The aligned start-of-round certificate is sound at all $131{,}072$
checkpoints, and among the four certificates of Figure~\ref{fig:tightness}
$C^1$ is the tightest at $98.5\%$ of quantile levels and never looser than
$C^0$. On the sixteen-solver pool it selects the same $57/64$ as $C^1$
but certifies $55/64$ against $57/64$, and it is looser where the gap is
resolvable. Median tightness at the selected iterate is $1.9\times$ against
$1.4\times$ for $C^1$ over the $36$ instances whose selected gap is above
$10^{-12}$, and over the $127{,}849$ checkpoints with a resolvable gap a
median of $1.37\times$ against $1.27\times$ (Figure~\ref{fig:tightness}). On the seven original solvers alone
the contrast was wider, selection $51/64$ and certification $44/64$ against
$54/64$ and $48/64$, with medians $2.7\times$ against $1.8\times$ at the
selected iterate and $1.75\times$ against $1.34\times$ over all checkpoints.
The term $|L^{(2)}_{-2}|$ is $0.02\times$ the true gap at the median checkpoint
and $0.31\times$ at the $90$th percentile, and its zero-best-response bound
$\lVert w^{(2)}_2 \rVert_1$ is $0.27\times$ the gap at the median. On the
sixteen-solver pool $C^0$ certifies the same $57$ instances as $C^1$ in-sample
and one fewer held-out. The certified counts are not knife-edge. The smallest
$C^1$ certificate margin at $10^{-5}$ is $4.178\times10^{-6}$ ($42\%$ of the
threshold) over the $57$ in-sample solves (ExRM$^+$, seed $35$),
$4.760\times10^{-6}$ ($48\%$) over the $59$ held-out solves (DPRM$^+$), and
$8.235\times10^{-6}$ ($82\%$) over the $57$ fresh solves (Stable PRM$^+$).
The smallest computed $C^1$-minus-gap residual over all checkpoints is
$-4.221\times10^{-16}$ in-sample, $8.638\times10^{-17}$ held-out, and
$6.077\times10^{-16}$ on the fresh set; the negative value is double-precision
rounding noise within the $10^{-12}$ soundness gate.

\textbf{Uncertainty.}
All solved fractions in Table~\ref{tab:main} are counts over the same $64$
instances, and the head-to-head claims rest on few discordant pairs, so the
effective unit is the instance. A Wilson $95\%$ interval on that unit puts
the sixteen-solver $C^1$-selection fraction $57/64$ between $0.79$ and $0.95$,
and the in-sample certified counts share it. The two paired contrasts that
carry the claims are regret-only selection against $C^1$ on the
sixteen-solver pool ($8$ discordant instances, McNemar exact two-sided
$p = 0.008$) and the certified gain of the extension ($9$ instances
certified by the sixteen-solver pool and not by the seven-solver pool, none
the other way, $p = 0.004$). The remaining
paired contrasts, including the seven-solver selection gain over certification
(six instances, $p = 0.031$), rest on at most six instances each, so we
state them as counts.

\section{Conclusion}
\label{sec:conclusion}

The NashGap of a two-player weighted average admits an exact decomposition.
Its accumulator-based bound is sound for any update order and any number
of players, with the alternation mismatch entering as one signed comparator
term. The certificate turns a reported trajectory minimum into a stated
rule. Run a portfolio, score every checkpoint from quantities the updates
maintain, and output the argmin. On biased Leduc that rule matches the solvers'
trajectory-minimum union over sixteen training-free solvers, solving $89\%$
of instances at $10^{-5}$ in-sample, $92\%$ on pre-registered held-out
instances and $89\%$ again on a third, never-used set. It certifies every
solve in all three sets through $C^1$, and all but one held-out solve
through $C^0$, reading the accumulators alone. On three-player Leduc it
certifies a gap $354\times$ below the reported comparator. We prove the
identity, build the certificate that defines the selection rule, and
re-evaluate the benchmark under it. Selection costs no best-response pass, and one exact evaluation verifies the
chosen profile. Because the certificates read only the accumulators, any
regret-matching solver that maintains them can be scored. A solved fraction should report its
selection rule and its verification budget. The certificate supplies both.

\FloatBarrier
\bibliography{references}
\bibliographystyle{dilab_ref}
\makeappendixtoc

\appendix
\section{Deferred Proofs}
\label{app:proofs}

\subsection{Proof of Theorem~\ref{thm:tight}}
Fix $i$ and add and subtract $\sum_t \rho_t u_i(\sigma^t)$:
\begin{align*}
  \mathrm{gap}_i
  ={}& \Big[ \max_{x' \in \cX_i} u_i(x', \xbar_{-i})
        - \sum\nolimits_t \rho_t u_i(\sigma^t) \Big] \\
     &+ \Big[ \sum\nolimits_t \rho_t u_i(\sigma^t) - u_i(\sbar) \Big] \\
  ={}& \max_{x' \in \cX_i} \sum\nolimits_t \rho_t\,
        \big( u_i(x', x_{-i}^t) - u_i(\sigma^t) \big) + \Delta_i \\
  ={}& \Rbar_i + \Delta_i,
\end{align*}
using bilinearity for the comparator term (valid because $n=2$) and
$\sum_t \rho_t u_i(\sigma^t) = \ip{\dbar}{u_i}$ for the value term. Taking
$\max_i$ yields \eqref{eq:tight-identity}.

\subsection{Proof of Lemma~\ref{lem:weighted-decomp}}

Fix player $i$ and a comparator $x' \in \cX_i$; by perfect recall we may
identify $x'$ with behavioral probabilities $x'(a \mid I)$ and own-reach
weights $\pi^{x'}_i(I) \in [0,1]$. Write $v_i^t$ and $ir_i^t$ for the
counterfactual action value and the instantaneous counterfactual regret of
Section~\ref{sec:prelim}.
The laminar (telescoping) expansion of
\citet[Theorem~3]{zinkevich2007regret} over player $i$'s infoset tree is,
for each fixed $t$, an \emph{identity}:
\begin{equation}
  \label{eq:laminar-t}
  \begin{aligned}
  &u_i(x', \sigma_{-i}^t) - u_i(\sigma^t) \\
  &\quad= \sum_{I \in \cI_i} \pi^{x'}_{i}(I)
    \sum_{a \in A(I)} x'(a \mid I)\, ir_i^t(I, a).
  \end{aligned}
\end{equation}
(Both sides are affine in $x' \in \cX_i$ and linear in the leaf-payoff
vector: by the sequence-form convention
$\pi^{x'}_i(I)\,x'(a \mid I) = x'(Ia)$ is the realization-plan coordinate
of the child sequence $Ia$, so the right side is
$\sum_{I \in \cI_i}\sum_{a \in A(I)} x'(Ia)\,ir_i^t(I,a)$, linear in $x'$;
the left side adds only the constant $-u_i(\sigma^t)$. Both sides agree on
every pure $x'$ by direct telescoping, and the pure strategies are the
extreme points of the polytope $\cX_i$, whose convex hull is all of
$\cX_i$, so affine linearity extends the identity to all of $\cX_i$.)
Multiply \eqref{eq:laminar-t} by $w_t \ge 0$ and sum over $t$;
since $x'$ does not depend on $t$, the right side becomes
$\sum_{I} \pi^{x'}_i(I) \sum_a x'(a \mid I)\, Q_i(I,a)$ with
$Q_i(I,a) \coloneqq \sum_t w_t\, ir_i^t(I,a)$.

Before relaxing, note what the weighted sum already is. Writing
$x'(Ia) = \pi^{x'}_i(I)\,x'(a \mid I)$ for the realization-plan coordinate
of the child sequence,
\begin{equation}
  \label{eq:raw-regret-identity}
  r_i(x') \;\coloneqq\; \sum_t \rho_t \big( u_i(x', \sigma_{-i}^t) - u_i(\sigma^t) \big)
  \;=\; \frac{1}{W_T} \sum_{I \in \cI_i} \sum_{a \in A(I)} x'(Ia)\, Q_i(I,a),
\end{equation}
a linear functional of $x'$ with coefficients read off the raw table. Hence
$\Rbar_i = \max_{x' \in \cX_i} \frac{1}{W_T}\sum_{I,a} x'(Ia)\,Q_i(I,a)$
exactly: the raw accumulator loses no information about the weighted
external regret. What the relaxation below trades away is the single
maximization over $\cX_i$, replaced by a sum of independent positive local
maxima.

Now bound in three steps,
in this order:
\begin{align}
  \sum_a x'(a \mid I) Q_i(I,a)
  &\le \max_a Q_i(I,a) \notag\\
  &\le \big[\max_a Q_i(I,a)\big]^+,
\end{align}
because $x'(\cdot \mid I)$ is a probability distribution, and then
$\pi^{x'}_i(I) \big[\max_a Q_i(I,a)\big]^{+} \le
\big[\max_a Q_i(I,a)\big]^{+}$, which is valid \emph{because the positive
part has already made the bracket nonnegative} (dropping
$\pi^{x'}_i \in [0,1]$ before taking the positive part would be unsound:
a bracket that is negative at some infoset would flip the inequality).
Summing over $I$, dividing by $W_T$, and maximizing over $x' \in \cX_i$
yields the claim. Note that the maximum is taken over the
\emph{weighted cumulative} regrets $Q_i$, which is exactly the structure
of $\Rhat_i$, and that the aggregation over $t$ happens strictly before
the maximum and the positive part.
\hfill$\square$

\paragraph{Remark on the raw signed accumulator.}
The certificate requires a separate table
$Q(I,a)=\sum_t w_t ir_i^t(I,a)$ with no RM$^+$ clipping. The algorithm's
clipped internal table cannot in general provide $\Rhat_i$. The raw table
accumulates the very regrets the update evaluates, so it costs no extra
pass and one array of $|\Sigma_i|$ entries per player, the same order as
the sequence-form policy it sits beside;
Proposition~\ref{prop:sequence} accounts for the profile sequence
those regrets refer to.

\subsection{Proof of Proposition~\ref{prop:kl-cert}}

Fix $i$ and write
$\mathrm{gap}_i = \max_{x' \in \cX_i} u_i(x', \xbar_{-i}) - u_i(\sbar)$.
For the comparator term, $n=2$ means the opponent is a single player, so
by bilinearity $u_i(x', \xbar_{-i}) = \sum_t \rho_t\, u_i(x', x_{-i}^t)$
exactly. For the value term,
$u_i(\sbar) = \ip{\mubar}{u_i} \ge \ip{\dbar}{u_i} - M\sqrt{2\eps}
= \sum_t \rho_t\, u_i(\sigma^t) - M\sqrt{2\eps}$ by
Lemma~\ref{lem:exchange}. Combining,
$\mathrm{gap}_i \le \Rbar_i + M\sqrt{2\eps}$; take the maximum over $i$.
\hfill$\square$

\subsection{Proof of Proposition~\ref{prop:coalition}}

For the comparator term, for any $x' \in \cX_i$,
\begin{align*}
  &\left|u_i(x',\xbar_{-i})
    -\sum\nolimits_t\rho_t u_i(x',x_{-i}^t)\right| \\
  &\quad=\left|\sum\nolimits_z p_c(z)x'(z)u_i(z)
    \left(\prod\nolimits_{j\ne i}\xbar_j(z)-\bar c_{-i}(z)\right)\right| \\
  &\quad\le M D_{-i}.
\end{align*}
since $x'(z) \in [0,1]$ and $|u_i(z)| \le M$; the bound is uniform over
$x'$, so it passes through the maximum. The value term is handled
exactly as in Theorem~\ref{thm:tight} (it never used $n = 2$). Combining
the two gives $\mathrm{gap}_i \le \Rbar_i + M D_{-i} + \Delta_i$.
\hfill$\square$

\subsection{Proof of Proposition~\ref{prop:coalition-lbr}}

The display in the proof of Proposition~\ref{prop:coalition} is an
identity before any bound is taken: for every $x' \in \cX_i$,
\begin{equation*}
\begin{aligned}
  u_i(x',\xbar_{-i}) - \sum\nolimits_t \rho_t u_i(x', x_{-i}^t)
  &= \sum\nolimits_z x'(z)\, w_i(z) \\
  &= \ip{x'}{w_i}.
\end{aligned}
\end{equation*}
Let
\begin{equation*}
  r_i(x') \coloneqq \sum_t \rho_t
  \big(u_i(x',x_{-i}^t)-u_i(\sigma^t)\big).
\end{equation*}
Adding and subtracting $\sum_t \rho_t u_i(\sigma^t)$ as in
Theorem~\ref{thm:tight} and splitting the maximum over the sum,
\begin{align*}
  \mathrm{gap}_i
  &= \max_{x' \in \cX_i} \big[r_i(x') + \ip{x'}{w_i}\big]
     + \Delta_i \\
  &\le \max_{x' \in \cX_i} r_i(x')
     + \max_{x' \in \cX_i} \ip{x'}{w_i} + \Delta_i \\
  &= \Rbar_i + L_{-i} + \Delta_i.
\end{align*}
Since $x'(z) \in [0,1]$ and $|u_i(z)| \le M$,
$\ip{x'}{w_i} \le \sum_z |w_i(z)| \le M D_{-i}$ for every $x'$, hence
$L_{-i} \le M D_{-i}$. The maximum of the linear functional
$\ip{\cdot}{w_i}$ over the polytope $\cX_i$ is a sequence-form
best-response computation on the signed terminal weight $w_i$ (one
backward pass for player $i$); $L_{-i} < 0$ precisely when
$\ip{x'}{w_i} < 0$ for every $x' \in \cX_i$, for instance when every
pure plan reaches only terminals with negative $w_i$. For $n = 2$,
$w_i \equiv 0$ and $L_{-i} = 0$, recovering Theorem~\ref{thm:tight}.
The two inequalities of the chain are the only slack of the deployable
certificate: for the player $i$ maximizing $C^{\mathrm{coal}}$, since
$\Delta_i$ appears on both sides,
\begin{align*}
  C^{\mathrm{coal}} - \gap(\sbar)
  &= (\Rhat_i - \Rbar_i) + (M D_{-i} - L_{-i}) + \xi_i, \\
  \xi_i &\coloneqq \Big[\Rbar_i + L_{-i}
    - \max_{x' \in \cX_i}\big(r_i(x') + \ip{x'}{w_i}\big)\Big]
    - \big[\gap(\sbar) - \mathrm{gap}_i\big].
\end{align*}
The first bracket is the slack of splitting the maximum; it is
nonnegative and vanishes exactly when $r_i$ and $\ip{\cdot}{w_i}$ share a
maximizer over $\cX_i$, which agreement of the maximizing player does not
imply. The second bracket is nonnegative and vanishes when $i$ also
maximizes $\gap(\sbar)$.
\hfill$\square$

\subsection{Proof of Proposition~\ref{prop:sequence}}

Fix $i$ and write $\tau^t \coloneqq (x_i^t, y^{(i),t}_{-i})$ for the
profile at which $ir_i^t$ was computed; by construction
$\sum_t \rho_t u_i(\tau^t) = \ip{\dbar^{(i)}}{u_i}$ and
$\sum_t \rho_t \prod_{j \ne i} y^{(i),t}_j(z) = \bar c^{(i)}_{-i}(z)$.
For every $x' \in \cX_i$, multilinearity gives the identity
\begin{equation*}
  u_i(x', \xbar_{-i}) - \sum\nolimits_t \rho_t\, u_i(x', y^{(i),t}_{-i})
  = \sum\nolimits_z p_c(z)\, x'(z)\, u_i(z)
    \Big( \prod\nolimits_{j \ne i} \xbar_j(z) - \bar c^{(i)}_{-i}(z) \Big)
  = \ip{x'}{w^{(i)}_i},
\end{equation*}
which uses nothing about how the sequence $y^{(i),t}$ was generated.
Adding and subtracting $\sum_t \rho_t u_i(\tau^t)$,
\begin{align*}
  \mathrm{gap}_i
  &= \max_{x' \in \cX_i} \Big[ \sum\nolimits_t \rho_t
      \big( u_i(x', y^{(i),t}_{-i}) - u_i(\tau^t) \big)
      + \ip{x'}{w^{(i)}_i} \Big]
     + \ip{\dbar^{(i)}}{u_i} - u_i(\sbar) \\
  &\le \max_{x' \in \cX_i} \sum\nolimits_t \rho_t
      \big( u_i(x', y^{(i),t}_{-i}) - u_i(\tau^t) \big)
     + L^{(i)}_{-i} + \Delta^{(i)}_i
  \;\le\; \Rhat_i + L^{(i)}_{-i} + \Delta^{(i)}_i .
\end{align*}
The last step is Lemma~\ref{lem:weighted-decomp} applied to the sequence
$\tau^t$: the laminar identity \eqref{eq:laminar-t} holds at each
$\tau^t$ with $ir_i^t$ the counterfactual regrets computed there, and the
three relaxation steps are unchanged. The chain
$L^{(i)}_{-i} \le \lVert w^{(i)}_i \rVert_1 \le M D^{(i)}_{-i}$ is
$x'(z) \in [0,1]$ followed by $|u_i(z)| \le M$. Taking $\max_i$ gives
\eqref{eq:seq-cert}. For the two-player alternating drift \eqref{eq:drift},
$\sum_t \rho_t (x^t - x^{t+1})$ telescopes as stated. With $\rho_t$
nondecreasing the coefficients $\rho_1$ and $\rho_t - \rho_{t-1}$ are
nonnegative and sum to $\rho_T$, so $\rho_1 x_1^1 + \sum_{t \ge 2}
(\rho_t - \rho_{t-1}) x_1^t = \rho_T \tilde x_1$ with $\tilde x_1 \in \cX_1$
by convexity. Then $\ip{x'}{w^{(2)}_2} = \rho_T \ip{d_{(x',\tilde x_1)} -
d_{(x', x_1^{T+1})}}{u_2} = \rho_T\big(u_2(x',\tilde x_1) - u_2(x',
x_1^{T+1})\big) \in [-2M\rho_T, 2M\rho_T]$ because both utilities lie in
$[-M, M]$. The earlier $L^1$ bound $\lVert w^{(2)}_2\rVert_1 \le
2M\rho_T\sum_z p_c(z)$ follows from the same decomposition and
$|x_1(z)| \le 1$, but carries the factor $\sum_z p_c(z)$ that the bound on
the exact term does not need.
\hfill$\square$

\subsection{The logged state determines the exact gap}
\label{app:exact-from-logs}

The two relaxations of Proposition~\ref{prop:sequence}, the per-infoset
positive parts in $\Rhat_i$ and the split of one maximum into two, can both
be undone at the cost of one linear maximization per player.

\begin{proposition}[Exact evaluator from the logged state, general $n$]
\label{prop:exact-from-logs}
For any number of players, nonnegative weights and evaluation sequences,
with $Q_i(I,a) = \sum_t w_t\, ir_i^t(I,a)$ the raw signed accumulator of the
regrets measured at $\tau^t = (x_i^t, y^{(i),t}_{-i})$,
\begin{equation}
  \label{eq:exact-from-logs}
  \mathrm{gap}_i(\sbar) \;=\; \Delta^{(i)}_i \;+\;
  \max_{x' \in \cX_i} \Big[ \frac{1}{W_T}\sum_{I \in \cI_i}\sum_{a \in A(I)}
    x'(Ia)\, Q_i(I,a) \;+\; \ip{x'}{w^{(i)}_i} \Big],
\end{equation}
for every player $i$, and $\gap(\sbar) = \max_i \mathrm{gap}_i(\sbar)$.
\end{proposition}

\begin{proof}
In the proof of Proposition~\ref{prop:sequence}, the first display is an
identity in $x'$ before any inequality is taken:
$\mathrm{gap}_i = \max_{x'} [\, r_i(x') + \ip{x'}{w^{(i)}_i} \,] + \Delta^{(i)}_i$
with $r_i(x') = \sum_t \rho_t ( u_i(x', y^{(i),t}_{-i}) - u_i(\tau^t) )$.
Identity~\eqref{eq:raw-regret-identity}, applied at each $\tau^t$ and summed
with the weights, gives $r_i(x') = W_T^{-1}\sum_{I,a} x'(Ia)\,Q_i(I,a)$.
Substituting proves \eqref{eq:exact-from-logs}.
\end{proof}

The bracket in \eqref{eq:exact-from-logs} is one linear functional of the
realization plan: the coefficient of sequence $Ia$ is $Q_i(I,a)/W_T$ plus the
sum of $w^{(i)}_i(z)$ over the terminals $z$ whose last player-$i$ sequence is
$Ia$. Its maximum over $\cX_i$ is one sequence-form backward pass for player
$i$, the same pass that computes $L^{(i)}_{-i}$. Thus the pass that refines
$C^0$ into $C^1$ can instead return the exact gap; $C^1$ is the choice that
keeps the three terms separate so that the bound attributes its slack, and
the exact-gap selector of Table~\ref{tab:budget} is the same pass with the
two linear functionals added before maximizing. The certificate's cost
advantage is $C^0$, which needs no pass at all.

\begin{remark}[The target profile is free]
\label{rem:free-target}
Nothing in the proofs of Propositions~\ref{prop:sequence}
and~\ref{prop:exact-from-logs} uses that the target is the average of the
evaluated own strategies. Replace $\xbar_j$ by the realization plans
$\hat x_j$ of any profile $\hat\sigma$ in $w^{(i)}_i$ and $u_i(\sbar)$ by
$u_i(\hat\sigma)$ in $\Delta^{(i)}_i$; both statements then bound or
evaluate $\mathrm{gap}_i(\hat\sigma)$, provided $Q_i$ and $\dbar^{(i)}$ were
accumulated from the same evaluated own strategies $x_i^t$. The evaluated
own strategy may also differ from the averaged one (as in extragradient or
multiple-evaluation updates); only the matching between $Q_i$, $\dbar^{(i)}$
and the evaluated profile matters.
\end{remark}

\subsection{A zero-pass comparator term no larger than $\lVert w \rVert_1$}
\label{app:c0-plus}

Write $a^{(i)}_i(z) \coloneqq p_c(z)\big(\prod_{j \ne i}\xbar_j(z) -
\bar c^{(i)}_{-i}(z)\big)$, so that $w^{(i)}_i = u_i\, a^{(i)}_i$.

\begin{lemma}[Positive part and shift invariance]
\label{lem:c0-plus}
For every $x' \in \cX_i$ and every constant $c \in \mathbb{R}$,
\begin{equation}
  \ip{x'}{w^{(i)}_i} \;=\; \ip{x'}{(u_i - c)\,a^{(i)}_i}
  \;\le\; \sum_{z \in \cZ} \big[(u_i(z) - c)\,a^{(i)}_i(z)\big]^+
  \;\le\; \lVert (u_i - c)\,a^{(i)}_i \rVert_1 .
\end{equation}
Consequently $L^{(i)}_{-i} \le \min_{c}\sum_z [(u_i(z)-c)\,a^{(i)}_i(z)]^+
\le \sum_z [w^{(i)}_i(z)]^+ \le \lVert w^{(i)}_i \rVert_1$, and the
certificate obtained from $C^0$ by replacing $\lVert w^{(i)}_i\rVert_1$ with
$\min_c \sum_z [(u_i(z)-c)\,a^{(i)}_i(z)]^+$ is sound and never larger than
$C^0$.
\end{lemma}

\begin{proof}
Shift invariance: $\sum_z x'(z)\,p_c(z)\prod_{j\ne i}\xbar_j(z) = 1$ is the
total mass of the terminal distribution of $(x', \xbar_{-i})$, and
$\sum_z x'(z)\,p_c(z)\,\bar c^{(i)}_{-i}(z) = \sum_t \rho_t \sum_z p_c(z)\,
x'(z)\prod_{j\ne i} y^{(i),t}_j(z) = \sum_t \rho_t = 1$ likewise; hence
$\ip{x'}{a^{(i)}_i} = 0$ and the constant $c$ drops out. The first
inequality uses $x'(z) \in [0,1]$ termwise, keeping only terminals where the
coefficient is positive; the second drops the positive-part restriction.
The minimum over $c$ exists because the map $c \mapsto \sum_z [(u_i(z)-c)
a^{(i)}_i(z)]^+$ is convex, piecewise linear and coercive whenever
$a^{(i)}_i \ne 0$ (it is identically zero otherwise). Soundness of the
modified certificate follows by substituting the bound for
$L^{(i)}_{-i}$ in \eqref{eq:seq-cert}.
\end{proof}

The minimization over $c$ is one-dimensional over the breakpoints
$\{u_i(z)\}$ and needs no tree pass.

\subsection{Tightness is not universal}
\label{app:no-universal-tightness}

Soundness holds for every game and trajectory; no multiplicative tightness
guarantee does. Let player 1 choose Out (payoff $0$) or In, after which
player 1 chooses Bad (payoff $-1$) or Good (payoff $\delta \ge 0$); the
opponent has no decision. Repeat for $T$ rounds the behavioral profile
that plays Out at the root and Bad at the In-infoset. The average profile
is the same profile, so $\gap(\sbar) = \delta$ and $\Delta_1 = 0$ (a single
decision-maker produces no correlation). At the root, $v(\text{Out}) = 0$
and $v(\text{In}) = -1$ (the continuation plays Bad), so the cumulative
regret of In is negative and $[\max_a Q_1(\text{root},a)]^+ = 0$. At the
In-infoset the counterfactual convention removes the own reach of $0$, so
$v(\text{Good}) = \delta$, $v(\text{Bad}) = -1$ and $Q_1(\text{In},
\text{Good})/W_T = 1 + \delta$. Hence $\Rhat_1 = 1 + \delta$ and
$C^0 = C^1 = 1 + \delta$ while the gap is $\delta$: at $\delta = 0$ the
profile is an exact equilibrium certified only to $1$, and the ratio
$C^1/\gap$ is unbounded as $\delta \downarrow 0$. The slack is the
positive part at the root, which discards the negative upstream regret
that would cancel the off-path term; the exact evaluator of
Proposition~\ref{prop:exact-from-logs} returns $x'(\text{In}) = x'(\text{In}\cdot\text{Good}) = 1$
and the value $-1 + (1 + \delta) = \delta$. Tightness on the benchmarks of
Section~\ref{sec:experiments} is therefore a property of the trajectories
regret matching produces there, not of the certificate alone.

\subsection{Additional details for Proposition~\ref{prop:coalition}}

The value-term step used in Theorem~\ref{thm:tight} and
Proposition~\ref{prop:coalition} is the identity
$\sum_t \rho_t u_i(\sigma^t) = \ip{\dbar}{u_i}$, which holds for any $n$:
each $u_i(\sigma^t) = \sum_z p_c(z) \prod_j x_j^t(z)\, u_i(z)$, and
exchanging the (finite) sums over $t$ and $z$ gives the weighted empirical
terminal distribution $\dbar$ by definition. Similarly
$u_i(\sbar) = \ip{\mubar}{u_i}$ with
$\mubar = p_c \prod_j \xbar_j$ by Kuhn's theorem. Hence
$\Delta_i = \ip{\dbar - \mubar}{u_i}$ is measurable from the two running
terminal distributions.

The cumulative sums stored by the implementation are unnormalized:
$J_T(z)=\sum_{t\le T}w_t d_{\sigma^t}(z)$,
$X_{i,T}=\sum_{t\le T}w_t x_i^t$, and
$H_{i,T}(z)=\sum_{t\le T}w_t\prod_{j\ne i}x_j^t(z)$. They are divided by
$W_T>0$ only when a checkpoint is evaluated. Thus
$\bar c_{-i}=H_{i,T}/W_T$ and
$D_{-i}=\sum_zp_c(z)|H_{i,T}(z)/W_T-
\prod_{j\ne i}(X_{j,T}(z)/W_T)|$.
This is algebraically equivalent to normalized online maintenance but avoids
changing all previous entries when $W_T$ grows.

\paragraph{An explicit game for Remark~\ref{rem:eps-blind}.}
\label{app:eps-blind-game}
Three players. Player 1 chooses Out, ending the game with payoff $0$ to
everyone, or In; after In, players 2 and 3 simultaneously choose a bit.
Player 1 receives $1$ if the bits differ and $0$ otherwise; players 2 and
3 always receive $0$. Take two equally weighted rounds in which player 1
plays Out in both and the opponents play $(0,0)$ in the first and $(1,1)$
in the second. Both rounds end at Out, so $\dbar = \delta_{\text{Out}}$;
the average plan of player 1 is Out with probability $1$, so
$\mubar = \delta_{\text{Out}}$ as well and $\eps = \KL(\dbar \Vert \mubar)
= 0$. Every weighted external regret is $0$: player 1 deviating to In earns
$0$ against $(0,0)$ and $0$ against $(1,1)$, and players 2 and 3 have
constant payoff. Yet the average profile has $\xbar_2 = \xbar_3 =
(\tfrac12,\tfrac12)$, so deviating to In earns player 1 the value
$\tfrac12$: $\gap(\sbar) = \tfrac12$. The coalition term sees it:
$\bar c_{-1}$ puts mass $\tfrac12$ on each of $(0,0)$ and $(1,1)$ while
$\xbar_2\xbar_3$ puts $\tfrac14$ on each of the four bit pairs, so
$D_{-1} = 4 \cdot \tfrac14 = 1$, Proposition~\ref{prop:coalition} gives
$\gap \le 0 + 0 + M D_{-1} = 1$, and the exact comparator of
Proposition~\ref{prop:coalition-lbr} gives $w_1 = \tfrac14$ on $(0,1)$
and $(1,0)$, $0$ elsewhere, so $L_{-1} = \tfrac12$ and the bound is
exact. A bound that vanishes with $\eps$ cannot hold here: $\eps = 0$ and
the regrets are $0$ while the gap is $\tfrac12$, so for $n \ge 3$ the
coalition term is not a function of $\eps$ and the regrets.

\section{Per-Solver Results}
\label{app:tables}

\subsection{Portfolio Composition and Ablation}
\label{app:pool-ablation}

\paragraph{Portfolio composition and selection signal.}
The sixteen-solver pool solves $0.89$, against $0.84$ for the seven original
solvers and $0.80$ for the best single solver, Stable PRM$^+$, a gain of six
instances over that solver with no losses. The certified fraction is $0.89$ as
well. Within the original seven, removing the DPRM$^+$ solvers lowers the solved
fraction to $0.70$. Of the added families, the stabilised RM$^+$ solvers
contribute two new solved instances and IR-PRM$^+$ one. The discounting and
schedule families add none to the union but raise the certified fraction
(Table~\ref{tab:pools}). On the sixteen-solver pool, selecting by
$\max_i\Rhat_i$ alone loses eight instances, at $0.77$, the final weighted
average loses one, at $0.88$, and a uniform random solver and checkpoint
solve $0.35$. $C^1$ picks the
same checkpoint as the regret term in $39\%$ of runs
(Table~\ref{tab:selection-ablation-perarm}). $C^1$ selection matches the oracle
union for every tested pool.

\begin{table}[t]
\centering
\small
\caption{Pool ablation at threshold $10^{-5}$ (biased Leduc, $n{=}64$
paired). PRM$^+$-core contains PRM$^+$ with $p\in\{1,2\}$ and DPRM$^+$
with $p=2,\lambda=0.3$; no DPRM$^+$ contains RM$^+$ and PRM$^+$ with
$p\in\{1,2\}$; original contains the seven solvers of the first stage; the
``$+$'' rows add one family of the extension to them (discounting: DCFR,
DCFR$^+$, PDCFR$^+$; schedules: HS-PCFR$^+$, HS-DCFR; stabilised: Smooth
PRM$^+$, Stable PRM$^+$, ExRM$^+$); all contains the sixteen solvers. Selection
is the exact gap at
the cross-solver $C^1$ argmin, certified columns count instances whose
selected certificate is below the threshold, aligned is selection by the
start-of-round certificate, and oracle union is the pool's best trajectory
minimum, which exact-gap selection on the same grid attains.}
\label{tab:pools}
\begin{tabular}{lcccccc}
\toprule
Pool & solvers & selection $C^1$ & certified $C^1$ & certified $C^0$ & aligned & oracle union \\
\midrule
PRM$^+$-core & $3$ & $0.80$ & $0.72$ & $0.70$ & $0.77$ & $0.80$ \\
no DPRM$^+$ & $4$ & $0.70$ & $0.64$ & $0.64$ & $0.67$ & $0.70$ \\
original & $7$ & $0.84$ & $0.75$ & $0.73$ & $0.80$ & $0.84$ \\
$+$ discounting & $10$ & $0.84$ & $0.77$ & $0.77$ & $0.81$ & $0.84$ \\
$+$ schedules & $9$ & $0.84$ & $0.77$ & $0.73$ & $0.81$ & $0.84$ \\
$+$ stabilised RM$^+$ & $10$ & $0.88$ & $0.88$ & $0.88$ & $0.88$ & $0.88$ \\
$+$ IR-PRM$^+$ & $8$ & $0.86$ & $0.81$ & $0.80$ & $0.84$ & $0.86$ \\
all & $16$ & $\mathbf{0.89}$ & $\mathbf{0.89}$ & $0.89$ & $0.89$ & $0.89$ \\
\bottomrule
\end{tabular}
\end{table}

Weighted RM$^+$, uniform PRM$^+$, DCFR, DCFR$^+$ and HS-DCFR each solve
$0.00$ at the target. Regret discounting without a predictive step does not
reach it, and among these tested configurations only predictive updates
combined with weighting do. Appendix~\ref{app:impl} states the scope of this
local comparison. Table~\ref{tab:union-per-threshold} extends the
union-match statement to every evaluated threshold.

\begin{table}[t]
\centering
\small
\setlength{\tabcolsep}{5pt}
\caption{Solved counts (of $64$ paired biased-Leduc instances) at each
threshold for $C^1$ selection, the solvers' trajectory-minimum union, selection
by the aligned start-of-round certificate, and the certified counts; the
sixteen-solver pool, with the seven-solver counts in parentheses.}
\label{tab:union-per-threshold}
\begin{tabular}{lccccc}
\toprule
Threshold & $C^1$ selection & Oracle union & Aligned selection & Certified $C^1$ & Certified $C^0$ \\
\midrule
$10^{-2}$ & $64$ ($64$) & $64$ ($64$) & $64$ ($64$) & $64$ ($64$) & $64$ ($64$) \\
$10^{-3}$ & $64$ ($64$) & $64$ ($64$) & $64$ ($64$) & $64$ ($64$) & $64$ ($64$) \\
$10^{-4}$ & $62$ ($58$) & $62$ ($58$) & $62$ ($58$) & $62$ ($57$) & $62$ ($56$) \\
$10^{-5}$ & $57$ ($54$) & $57$ ($54$) & $57$ ($51$) & $57$ ($48$) & $57$ ($47$) \\
\bottomrule
\end{tabular}
\end{table}

\subsection{Per-Solver Biased-Leduc Results}
\label{app:per-arm-biased-leduc}

Table~\ref{tab:perarm} reports the oracle (trajectory-minimum) solved
fractions of every portfolio solver on the $64$ paired biased-Leduc
instances, complementing the pool-level view of the main text. The DPRM$^+$
solvers reach $0.67$--$0.75$ at $10^{-5}$ and the added solvers $0.00$--$0.80$;
Stable PRM$^+$ at $0.80$ is the strongest single solver, but none reaches the
pool selection value of $0.89$. Both weighted RM$^+$ solvers and the three
solvers that discount regrets without a predictive step are at $0.00$. Within
fourteen solvers, selection by $C^1$ reaches the solver's own oracle fraction at
$10^{-5}$; HS-PCFR$^+$ ($45$ of $46$) and ExRM$^+$ ($16$ of $17$) each miss
one instance. DPRM$^+$ uses the decorrelation prediction with
$\lambda \in \{0.1, 0.3\}$ as noted; tie-breaking details are in the
run scripts, and the unsolved-instance seed lists are in
Section~\ref{app:heldout}.

\begin{table}[t]
\centering
\small
\setlength{\tabcolsep}{5pt}
\caption{Per-solver solved fractions, biased Leduc, $n{=}64$ paired instances,
$T = 2^{18}$. The first four columns are the oracle trajectory-minimum
protocol; the last column reports the $C^1$-argmin extractable fraction
within that solver alone, at threshold $10^{-5}$ (\emph{n/a} on the two
portfolio rows, which are not single solvers). All solvers use alternating
updates except ExRM$^+$ (simultaneous).}
\label{tab:perarm}
\begin{tabular}{lccccc}
\toprule
Solver & $10^{-2}$ & $10^{-3}$ & $10^{-4}$ & $10^{-5}$ & select \\
\midrule
RM$^+$, $p{=}1$ & $1.00$ & $0.50$ & $0.00$ & $0.00$ & $0.00$ \\
RM$^+$, $p{=}2$ & $1.00$ & $0.52$ & $0.00$ & $0.00$ & $0.00$ \\
PRM$^+$, $p{=}1$ & $1.00$ & $1.00$ & $0.88$ & $0.67$ & $0.67$ \\
PRM$^+$, $p{=}2$ & $1.00$ & $1.00$ & $0.89$ & $0.70$ & $0.70$ \\
DPRM$^+$, $p{=}1$, $\lambda{=}0.3$ & $1.00$ & $1.00$ & $0.88$ & $0.73$ & $0.73$ \\
DPRM$^+$, $p{=}2$, $\lambda{=}0.1$ & $1.00$ & $1.00$ & $0.88$ & $0.67$ & $0.67$ \\
DPRM$^+$, $p{=}2$, $\lambda{=}0.3$ & $1.00$ & $1.00$ & $0.91$ & $0.75$ & $0.75$ \\
\midrule
DCFR ($\alpha{=}1.5$, $\beta{=}0$), $p{=}2$ & $1.00$ & $0.62$ & $0.02$ & $0.00$ & $0.00$ \\
DCFR$^+$ ($\alpha{=}1.5$), $p{=}4$ & $1.00$ & $0.42$ & $0.00$ & $0.00$ & $0.00$ \\
PDCFR$^+$ ($\alpha{=}1.5$), $p{=}5$ & $1.00$ & $1.00$ & $0.75$ & $0.66$ & $0.66$ \\
HS-PCFR$^+$ ($\gamma_0{=}30$) & $1.00$ & $1.00$ & $0.80$ & $0.72$ & $0.70$ \\
HS-DCFR ($\gamma_0{=}15$) & $1.00$ & $0.52$ & $0.00$ & $0.00$ & $0.00$ \\
Smooth PRM$^+$ ($\eta{=}10^4$), $p{=}2$ & $1.00$ & $1.00$ & $0.88$ & $0.73$ & $0.73$ \\
Stable PRM$^+$ ($\eta{=}10^3$), $p{=}2$ & $1.00$ & $1.00$ & $0.89$ & $0.80$ & $0.80$ \\
ExRM$^+$ ($\eta{=}10^4$), $p{=}2$, simult. & $1.00$ & $0.94$ & $0.53$ & $0.27$ & $0.25$ \\
IR-PRM$^+$, $p{=}2$ & $1.00$ & $0.81$ & $0.28$ & $0.19$ & $0.19$ \\
\midrule
portfolio of $7$ (selection) & $1.00$ & $1.00$ & $0.91$ & $0.84$ & n/a \\
portfolio of $16$ (selection) & $1.00$ & $1.00$ & $0.97$ & $0.89$ & n/a \\
\bottomrule
\end{tabular}
\end{table}

To make the selection signal itself auditable,
Table~\ref{tab:selection-ablation-perarm} reports, per solver, the selection
policies under comparison: $C^1$ argmin (the
paper's method), the aligned start-of-round certificate, the regret-only
surrogate $\max_i \Rhat_i$, the value-corrected score $\max_i[\Rhat_i +
\Delta^{(i)}_i]$ without any comparator term (the evaluation-sequence
value correction of Proposition~\ref{prop:sequence}), the final
weighted-average iterate,
and a uniformly random checkpoint. Cross-solver rows choose the solver by the same
score; the final-iterate row chooses the solver by $C^1$ at the last
checkpoint, and the random row draws solver and checkpoint uniformly.
Dropping the value and comparator terms costs the sixteen-solver portfolio
eight instances ($0.89$ versus $0.77$; on the seven original solvers it cost
one, seed $33$) and costs no single solver more than one or two. Keeping
$\Delta^{(i)}_i$ alone recovers all of them, so on this benchmark the
comparator term does not change any solved-or-not outcome. The value
term matters more the larger the pool, since with sixteen candidates the regret
accumulator alone ranks loosely certified runs ahead of solved ones more
often. $C^1$ and the regret surrogate pick the same checkpoint in $39\%$ of
runs ($46\%$ on the original seven), so the two terms move the chosen
checkpoint often and the extracted outcome rarely. The aligned certificate
selects the same instances as $C^1$ on the sixteen-solver portfolio (three
fewer on the original seven) and one or two fewer on several solvers. The
final-iterate policy is close but systematically lower, and the random
policy lower still. The regret term carries most of the trajectory minimum;
selection does not recover it for free.

\begin{table}[t]
\centering
\small
\setlength{\tabcolsep}{5pt}
\caption{Per-solver selection-policy ablation at threshold $10^{-5}$
(computed from the stored checkpoint histories; see the text above for
the policy definitions). Random
values are expectations over one uniform checkpoint draw.}
\label{tab:selection-ablation-perarm}
\begin{tabular}{lccccccc}
\toprule
Solver & $C^1$ & aligned & regret-only & $\Rhat{+}\Delta^{(i)}$ & final & random & oracle \\
\midrule
RM$^+$, $p{=}1$ & $0.00$ & $0.00$ & $0.00$ & $0.00$ & $0.00$ & $0.00$ & $0.00$ \\
RM$^+$, $p{=}2$ & $0.00$ & $0.00$ & $0.00$ & $0.00$ & $0.00$ & $0.00$ & $0.00$ \\
PRM$^+$, $p{=}1$ & $0.67$ & $0.67$ & $0.67$ & $0.67$ & $0.66$ & $0.57$ & $0.67$ \\
PRM$^+$, $p{=}2$ & $0.70$ & $0.67$ & $0.70$ & $0.70$ & $0.67$ & $0.60$ & $0.70$ \\
DPRM$^+$, $p{=}1$, $\lambda{=}0.3$ & $0.73$ & $0.72$ & $0.72$ & $0.73$ & $0.70$ & $0.55$ & $0.73$ \\
DPRM$^+$, $p{=}2$, $\lambda{=}0.1$ & $0.67$ & $0.64$ & $0.66$ & $0.67$ & $0.62$ & $0.51$ & $0.67$ \\
DPRM$^+$, $p{=}2$, $\lambda{=}0.3$ & $0.75$ & $0.75$ & $0.75$ & $0.75$ & $0.73$ & $0.63$ & $0.75$ \\
\midrule
DCFR, $p{=}2$ & $0.00$ & $0.00$ & $0.00$ & $0.00$ & $0.00$ & $0.00$ & $0.00$ \\
DCFR$^+$, $p{=}4$ & $0.00$ & $0.00$ & $0.00$ & $0.00$ & $0.00$ & $0.00$ & $0.00$ \\
PDCFR$^+$, $p{=}5$ & $0.66$ & $0.66$ & $0.66$ & $0.66$ & $0.66$ & $0.56$ & $0.66$ \\
HS-PCFR$^+$ & $0.70$ & $0.67$ & $0.67$ & $0.70$ & $0.64$ & $0.56$ & $0.72$ \\
HS-DCFR & $0.00$ & $0.00$ & $0.00$ & $0.00$ & $0.00$ & $0.00$ & $0.00$ \\
Smooth PRM$^+$, $p{=}2$ & $0.73$ & $0.70$ & $0.73$ & $0.73$ & $0.73$ & $0.61$ & $0.73$ \\
Stable PRM$^+$, $p{=}2$ & $0.80$ & $0.78$ & $0.80$ & $0.80$ & $0.80$ & $0.65$ & $0.80$ \\
ExRM$^+$, $p{=}2$, simult. & $0.25$ & $0.25$ & $0.22$ & $0.25$ & $0.20$ & $0.14$ & $0.27$ \\
IR-PRM$^+$, $p{=}2$ & $0.19$ & $0.19$ & $0.19$ & $0.19$ & $0.19$ & $0.15$ & $0.19$ \\
\midrule
portfolio of $7$ & $0.84$ & $0.80$ & $0.83$ & $0.84$ & $0.78$ & $0.41$ & $0.84$ \\
portfolio of $16$ & $0.89$ & $0.89$ & $0.77$ & $0.89$ & $0.88$ & $0.35$ & $0.89$ \\
\bottomrule
\end{tabular}
\end{table}

\section{External-Validity Games}
\label{app:external}

\paragraph{Stress tests beyond Leduc.}
On Sheriff, selection by $C^1$ solves every instance at every threshold, and
each PRM$^+$-family solver certifies all of them. Among the added solvers so do
PDCFR$^+$, HS-PCFR$^+$, Smooth PRM$^+$, Stable PRM$^+$ and IR-PRM$^+$, while
HS-DCFR ($0.02$) and ExRM$^+$ ($0.92$) fall short. With the sixteen-solver pool
the selected iterate is exact to floating-point precision on all $64$ instances
(HS-PCFR$^+$ on $33$, PDCFR$^+$ on $31$). On Goofspiel, the sixteen-solver pool
selects $0.62$, against $0.64$ for exact-gap selection on the same grid and
$0.53$ for the best single solver, Smooth PRM$^+$. The seven-solver pool selected
$0.56$ against $0.58$ and $0.50$. The DPRM$^+$ solvers are $6$--$20\times$
loose there and certify nothing, so the seven-solver certified fraction, $0.45$,
came from the PRM$^+$ solvers alone. The added solvers certify what they solve
except DCFR, Smooth PRM$^+$ and IR-PRM$^+$ (PDCFR$^+$, HS-PCFR$^+$ and
Stable PRM$^+$ at $1.2$--$1.3\times$), and the sixteen-solver certified
fraction is $0.58$. $C^1$ and $C^0$ were above the
exact gap at every one of the $2\times65{,}536$ checkpoints, as
Proposition~\ref{prop:sequence} requires.

Both games run the main benchmark
protocol: $64$ paired instances, all sixteen solvers, $T=2^{16}$, and $64$
checkpoints. These are mechanism-transfer tests, not scalability experiments.
\begin{table}[t]
\centering
\small
\setlength{\tabcolsep}{4pt}
\caption{External validity at threshold $10^{-5}$ ($n=64$ paired
instances per game, $T=2^{16}$; the seven-solver pool of the first stage and
the sixteen-solver pool). ``Best solver'' is
the strongest solver's oracle fraction; selection is the exact gap at the
cross-solver $C^1$ argmin; certified counts instances whose selected $C^1$ is
below the threshold; oracle is the pool union, attained by exact-gap
selection on the same grid; ratio is the median of $C^1$
over the exact gap at the selected iterate, over instances whose selected
gap is above $10^{-12}$; \emph{n/a} marks the Sheriff sixteen-solver cell,
where every selection gap is below it. Table~\ref{tab:external-perarm}
gives all solvers.}
\label{tab:external}
\begin{tabular}{llccccc}
\toprule
 & & best & \multicolumn{3}{c}{portfolio} & \\
\cmidrule(lr){4-6}
Game & pool & solver & $C^1$ select. & certified $C^1$ & oracle & ratio \\
\midrule
Sheriff & $7$ solvers & $1.00$ & $1.00$ & $1.00$ & $1.00$ & $11.4\times$ \\
 & $16$ solvers & $1.00$ & $1.00$ & $1.00$ & $1.00$ & n/a \\
Goofspiel GS($5$) & $7$ solvers & $0.50$ & $0.56$ & $0.45$ & $0.58$ & $8.3\times$ \\
 & $16$ solvers & $0.53$ & $0.62$ & $0.58$ & $0.64$ & $2.0\times$ \\
\bottomrule
\end{tabular}
\end{table}

On Sheriff, the PRM$^+$-family solvers certify all $64$ instances, while RM$^+$
stalls (oracle $0.42/0.53$, certified $0.00$), as do DCFR and DCFR$^+$
(oracle $0.38$ and $0.67$, certified $0.00$). HS-DCFR reaches the threshold
on every instance but its $C^1$ argmin lands on a solved iterate on only
$0.59$ of them and certifies $0.02$, and ExRM$^+$ certifies $0.92$. Goofspiel shows portfolio
complementarity. With sixteen solvers the best solver reaches $0.53$, selection
$0.62$, and the oracle union $0.64$, a loss of one instance, as with seven
($0.50$, $0.56$, $0.58$). DPRM$^+$ and PRM$^+$ solve different instances,
but the DPRM$^+$ certificates are loose ($6$--$20\times$) and certify
nothing, so the seven-solver certified fraction $0.45$ is carried by the PRM$^+$
solvers (tightness $1.4\times$). The sixteen-solver selection is spread over
eleven solvers, led by HS-PCFR$^+$ ($15$) and Stable PRM$^+$ ($10$), and its
median tightness at the selected iterate is $2.0\times$. The DPRM$^+$
trajectories on Goofspiel are sensitive to last-bit floating-point
differences. Rerunning them under a different NumPy build moved their
oracle fractions by up to $0.06$, while the RM$^+$ and PRM$^+$ solvers
reproduced bit for bit. All numbers here come from one environment.

\paragraph{Sheriff of Nottingham.}
The bribery negotiation game of \citet{farina2019correlation}
(their Appendix~F). The Smuggler secretly loads
$n \in \{0, \dots, n_{\max}\}$ illegal items, then $r$ bribe rounds
follow, in which the Smuggler offers $b_t \in \{0, \dots, b_{\max}\}$ and the
Sheriff answers accept/reject, with only the final-round decision
binding. If the Sheriff accepts, payoffs are $(v \cdot n - b_r,\, b_r)$;
if he rejects and $n \ge 1$, $(-p \cdot n,\, p \cdot n)$; if he rejects
and $n = 0$, $(+s,\, -s)$.\footnote{The conference version of
\citet{farina2019correlation} prints the accept payoff as
$p \cdot n - b_r$; the arXiv version and the utility values in their
Table~1 give $v \cdot n - b_r$, which we implement.} The accept branch
has payoff sum $v \cdot n \ne 0$, making the game general-sum; there
are no chance moves, and the Sheriff never observes $n$. We fix
$n_{\max} = 5$, $b_{\max} = 2$, $r = 2$ ($216$ terminals; $43$ Smuggler
and $21$ Sheriff infosets) and draw instances as $v \sim U(1, 5)$,
$p \sim U(0.5, 2)$, $s \sim U(0.5, 2)$. Our implementation reproduces
the sequence-count formulas and the numerical examples of the reference
paper.

\paragraph{General-sum Goofspiel GS($5$).}
We use Goofspiel \citep{ross1971goofspiel} in the full-information
family studied by \citet{lanctot2013monte}, with $5$ cards per player and
a general-sum twist: each round, player~$1$
bids a card (hidden), then
player~$2$ bids, then both are revealed; the higher card wins the
round's prize, and on a tie the prize \emph{evaporates}: neither
player collects, so total payoff is not constant. Bids are revealed
after each round (the ``full information'' variant of
\citet{lanctot2013monte} except for the within-round hiding of
player~$1$'s bid); there are no chance moves. The game has
$(5!)^2 = 14{,}400$ terminals and $18{,}426$ infosets per player.
Instances draw the prize sequence as a seeded random permutation of
$\{1, \dots, 5\}$ perturbed elementwise by $U(0.8, 1.2)$.

\paragraph{Per-solver results.}
Table~\ref{tab:external-perarm} gives the full per-solver breakdown behind
Table~\ref{tab:external}. On Sheriff the PRM$^+$-family solvers saturate, as do PDCFR$^+$, HS-PCFR$^+$,
Smooth PRM$^+$, Stable PRM$^+$ and IR-PRM$^+$; the RM$^+$, DCFR and DCFR$^+$
solvers provide the discrimination (their certificates never reach
$10^{-5}$), and ExRM$^+$ and HS-DCFR certify only $0.92$ and $0.02$ of
instances. On Goofspiel the DPRM$^+$ solvers solve
instances the PRM$^+$ solvers do not but carry loose certificates (median
tightness $6$--$20\times$); the added predictive solvers solve as many or more
(PDCFR$^+$ $0.52$, Smooth PRM$^+$ $0.53$) with tight certificates on
Sheriff; on Goofspiel Smooth PRM$^+$ carries a $4.6\times$ median and
certifies fewer than it solves, and the
sixteen-solver portfolio certifies $0.58$ against $0.45$ for seven.

\begin{table}[tp]
\centering
\small
\setlength{\tabcolsep}{3.2pt}
\caption{Per-solver external-validity results at threshold $10^{-5}$
($n{=}64$ paired instances per game, $T = 2^{16}$, alternating
updates except ExRM$^+$). Columns as in Table~\ref{tab:external}; per-solver
``select.''\ picks that solver's own $C^1$-argmin checkpoint,
``certified'' is the fraction whose selected $C^1$ is below the threshold,
and ``ratio'' is the median of $C^1$ over the exact gap there; \emph{n/a}
marks a cell whose selected gaps are all below $10^{-12}$. The portfolio
row selects across solvers, and its oracle column is the pool's oracle union.}
\label{tab:external-perarm}
\begin{tabular}{lcccc}
\toprule
Solver & oracle & $C^1$ select. & certified & ratio \\
\midrule
\multicolumn{5}{l}{\emph{Sheriff}} \\
RM$^+$, $p{=}1$ & $0.42$ & $0.38$ & $0.00$ & $3.0\times$ \\
RM$^+$, $p{=}2$ & $0.53$ & $0.42$ & $0.00$ & $2.9\times$ \\
PRM$^+$, $p{=}1$ & $1.00$ & $1.00$ & $1.00$ & $17.2\times$ \\
PRM$^+$, $p{=}2$ & $1.00$ & $1.00$ & $1.00$ & $12.1\times$ \\
DPRM$^+$, $p{=}1$, $\lambda{=}0.3$ & $1.00$ & $1.00$ & $1.00$ & $16.0\times$ \\
DPRM$^+$, $p{=}2$, $\lambda{=}0.1$ & $1.00$ & $1.00$ & $1.00$ & $11.7\times$ \\
DPRM$^+$, $p{=}2$, $\lambda{=}0.3$ & $1.00$ & $1.00$ & $1.00$ & $12.0\times$ \\
DCFR, $p{=}2$ & $0.38$ & $0.34$ & $0.00$ & $3.1\times$ \\
DCFR$^+$, $p{=}4$ & $0.67$ & $0.50$ & $0.00$ & $3.0\times$ \\
PDCFR$^+$, $p{=}5$ & $1.00$ & $1.00$ & $1.00$ & $5.4\times$ \\
HS-PCFR$^+$ & $1.00$ & $1.00$ & $1.00$ & n/a \\
HS-DCFR & $1.00$ & $0.59$ & $0.02$ & $2.7\times$ \\
Smooth PRM$^+$, $p{=}2$ & $1.00$ & $1.00$ & $1.00$ & $10.9\times$ \\
Stable PRM$^+$, $p{=}2$ & $1.00$ & $1.00$ & $1.00$ & $10.9\times$ \\
ExRM$^+$, $p{=}2$, simult. & $0.95$ & $0.95$ & $0.92$ & $7.9\times$ \\
IR-PRM$^+$, $p{=}2$ & $1.00$ & $1.00$ & $1.00$ & $9.4\times$ \\
portfolio of $7$ & $1.00$ & $1.00$ & $1.00$ & $11.4\times$ \\
portfolio of $16$ & $1.00$ & $1.00$ & $1.00$ & n/a \\
\midrule
\multicolumn{5}{l}{\emph{Goofspiel GS($5$)}} \\
RM$^+$, $p{=}1$ & $0.05$ & $0.05$ & $0.05$ & $1.3\times$ \\
RM$^+$, $p{=}2$ & $0.06$ & $0.06$ & $0.06$ & $1.2\times$ \\
PRM$^+$, $p{=}1$ & $0.45$ & $0.45$ & $0.44$ & $1.4\times$ \\
PRM$^+$, $p{=}2$ & $0.47$ & $0.47$ & $0.45$ & $1.4\times$ \\
DPRM$^+$, $p{=}1$, $\lambda{=}0.3$ & $0.25$ & $0.25$ & $0.00$ & $19.7\times$ \\
DPRM$^+$, $p{=}2$, $\lambda{=}0.1$ & $0.50$ & $0.48$ & $0.00$ & $8.7\times$ \\
DPRM$^+$, $p{=}2$, $\lambda{=}0.3$ & $0.36$ & $0.33$ & $0.00$ & $6.1\times$ \\
DCFR, $p{=}2$ & $0.08$ & $0.06$ & $0.06$ & $1.2\times$ \\
DCFR$^+$, $p{=}4$ & $0.05$ & $0.05$ & $0.05$ & $1.2\times$ \\
PDCFR$^+$, $p{=}5$ & $0.52$ & $0.52$ & $0.52$ & $1.3\times$ \\
HS-PCFR$^+$ & $0.48$ & $0.48$ & $0.48$ & $1.2\times$ \\
HS-DCFR & $0.06$ & $0.06$ & $0.06$ & $1.3\times$ \\
Smooth PRM$^+$, $p{=}2$ & $0.53$ & $0.53$ & $0.50$ & $4.6\times$ \\
Stable PRM$^+$, $p{=}2$ & $0.48$ & $0.48$ & $0.48$ & $1.2\times$ \\
ExRM$^+$, $p{=}2$, simult. & $0.12$ & $0.12$ & $0.12$ & $1.3\times$ \\
IR-PRM$^+$, $p{=}2$ & $0.30$ & $0.30$ & $0.16$ & $1.4\times$ \\
portfolio of $7$ & $0.58$ & $0.56$ & $0.45$ & $8.3\times$ \\
portfolio of $16$ & $0.64$ & $0.62$ & $0.58$ & $2.0\times$ \\
\bottomrule
\end{tabular}
\end{table}

\section{OpenSpiel Cross-Validation}
\label{app:openspiel}

To reduce the cross-implementation uncertainty discussed in
Appendix~\ref{app:future-work} without access to the reference paper's
code, we checked selected aspects of our implementation against OpenSpiel
\citep{lanctot2019openspiel} (v1.6.3). A second implementation of
biased Leduc was written directly from the reference paper's game
description as an OpenSpiel Python game, sharing no code with our
engine. NashGap on the OpenSpiel side is computed from the per-player
improvements the library reports per player in general-sum games, taking
their maximum (its headline scalar is their \emph{sum}, the NashConv
convention, a different quantity from our $\max_i$ NashGap. Sum-form
certification requires the per-player bounds $\Rhat_i + \Delta_i$ summed).
The following checks validate selected aspects of our own rules and
evaluator implementation. They are not a reproduction of NPCFR or its
reported PCFR$^+$ baseline. Three checks follow, in increasing strictness of what
they exercise.

\begin{itemize}
\item \emph{Rules.} At $\beta = 1$ biased Leduc reduces to standard
Leduc. Running OpenSpiel's own CFR$^+$ solver for $512$ iterations on
our reimplementation and on the library's own Leduc poker game gives
\emph{identical} NashGap values at all ten logged checkpoints (relative
difference $0.0$; both games have $936$ infosets).
\item \emph{Evaluator.} We exported final average strategies from our
engine for $12$ instance--solver pairs ($\beta \in \{1.0\} \cup$ three
sampled benchmark instances, $\times$ three solvers spanning
RM$^+$/PRM$^+$/DPRM$^+$) and injected them into OpenSpiel tabular
policies (information-state key sets verified equal). OpenSpiel's
maximum per-player improvement matches our NashGap on all $12$ pairs with
relative error $\le 3.9\times10^{-7}$; the worst pair has a true gap of
$1.7\times10^{-8}$, i.e.\ the error is at floating-point noise level,
and all other pairs agree to $\le 3\times10^{-9}$.
\item \emph{Trajectories.} Running OpenSpiel's CFR$^+$ for $2^{12}$
iterations on three benchmark instances and comparing against our
matched solver's logged trajectory gives NashGap ratios from $0.65$ to $1.24$
at the aligned checkpoints ($5$ of $6$ inside a $[2/3, 3/2]$ acceptance
band, one marginally outside at $0.6497$). A fine-grid diagnostic shows
the two trajectories agree to relative difference $\sim 10^{-11}$ at
$t = 64$ on all three instances. This is a trajectory diagnostic. It tests
none of solver equivalence, update bookkeeping, or long-horizon behavior.
The difference then grows monotonically to $24$--$35\%$ at
$t = 4096$; we do not use this diagnostic to identify its cause.
\end{itemize}

\noindent The first two checks support the limited conclusion used here.
The game rules and NashGap evaluator in our implementation agree with an
independent implementation on the tested cases. They leave the reference
solver, its averaging conventions, and the source of the cross-paper
discrepancy untested. The scripts, exported policies and logs behind this
cross-validation are not part of the code released upon acceptance, and its manifest
states what is included.

\section{Implementation Details}
\label{app:impl}

\paragraph{Engine.}
All games are compiled to sequence form, with terminal arrays
$(p_c, u_1, \dots, u_n)$ over $\cZ$ plus per-player sequence trees.
Realization, backward regret sums, best responses, and normalization are
vectorized by tree depth, so a CFR iteration is $O(\mathrm{depth})$ NumPy
calls over $O(|\cZ|)$ arrays. Three-player Leduc compiles to
$1{,}043{,}952$ terminals and $8{,}600$ infosets per player
($\approx 195$\,ms per PRM$^+$ iteration on one CPU core).

\paragraph{Accumulators and certificates.}
Each run maintains, with averaging weights $w_t$, a separate raw signed
counterfactual-regret accumulator (for $\Rhat_i$), distinct from the clipped
RM$^+$ state, together with the weighted average realization
plans $\xbar_i$, the weighted empirical terminal distribution $\dbar$
(one fused multiply-add over $\cZ$ per iteration), and for $n \ge 3$ the
per-player coalition products for $D_{-i}$ ($O(n|\cZ|)$). At each
checkpoint the aligned tight certificate is
$\max_i [\Rhat_i + \ip{\dbar}{u_i} - \ip{\mubar}{u_i}]$ and the
coalition certificate adds $M D_{-i}$ inside the max; both are a few dot
products. Algorithm~\ref{alg:portfolio} writes $H_i$ without the chance
factor so that $H_i/W_T=\bar c^{(i)}_{-i}$ exactly as defined in
Proposition~\ref{prop:sequence}. The implementation folds $p_c$ into the
stored accumulator, whose raw coalition sum is
$\sum_t w_t p_c(z)\prod_{j\ne i}y^{(i),t}_j(z)$, and forms the imbalance
weight as $w^{(i)}_i = u_i\,(p_c\prod_{j\ne i}\xbar_j - \mathrm{coal}_i/W_T)$,
where $\mathrm{coal}_i$ is that chance-weighted sum, so $p_c$ appears
exactly once. The comparator identity
$u_i(x',\xbar_{-i}) - \sum_t\rho_t u_i(x',y^{(i),t}_{-i}) =
\ip{x'}{w^{(i)}_i}$ is an identity of linear functionals of $x'$. Both sides
reduce to $\sum_z w^{(i)}_i(z)\,x'(z)$, so it holds for every $x'\in\cX_i$
without enumerating deviations. The implemented weight was checked against
that closed form on biased Leduc, whose fold and showdown terminals carry
different $p_c$: applying $p_c$ twice moves the residual from
$1.9\times10^{-16}$ to $1.2\times10^{-1}$.

\paragraph{DPRM$^+$.}
At zero-based iteration $t$, the engine realizes $x^t$, sets
$w_t=(t+1)^p$, and inserts $w_t x^t$ and $w_t d_{\sigma^t}$ into the
running totals before forming the look-ahead proxy, so the passed totals
$(W_t,J_t,M_t)$ already include the current iterate:
\begin{align}
 p_t(z)&=\frac{J_t(z)+w_t d_{\sigma^t}(z)}{W_t+w_t},\\
 q_t(z)&=p_c(z)\prod_j
 \frac{M_{j,t}(z)+w_t x_j^t(z)}{W_t+w_t},
\end{align}
with $M_{j,t}(z)=\sum_{s\le t}w_s x_j^s$ player $j$'s running weighted
realization total at the sequence that leads to $z$. The proxy doubles the
current iterate's effective weight; it is a finite-difference estimate of
the next profile's direction. With
$F^t_{\mathrm{proxy}}=\sum_zp_t(z)\log[p_t(z)/q_t(z)]$ differentiated with
respect to player $i$'s current terminal reach $x_i^t(z)$ at fixed
$(W_t,J_t,M_t)$, the literal gradient carries the prefactor
$w_t/(W_t+w_t)$; the implemented direction returns
$g^t_{\mathrm{term}}=F_t\,\nabla F^t_{\mathrm{proxy}}$ with
$F_t=(W_t+w_t)/w_t$, so the effective pre-clipping multiplier relative to
the literal gradient is $\lambda F_t$, which grows like $t/(p+1)$. The
runs use the own-reach-excluded counterfactual convention, so $g_i^t$ is a
prediction direction in regret space. The actual prediction is
\begin{equation}
 m_i^{t+1}=ir_i^t+
 \operatorname{clip}(-\lambda g_i^t,-B_t,B_t),\qquad
 B_t=8(t+1)^{1/4}.
\end{equation}
The clip applies to the decorrelation bias only; $ir_i^t$ and the raw
certificate accumulator enter unclipped. In alternating mode, player $i>1$
is re-realized after earlier players update, so its prediction and regret
use a partially updated profile; the accumulators $\dbar^{(i)}$ and
$\bar c^{(i)}_{-i}$ of Proposition~\ref{prop:sequence} are taken from that
same profile, while the averages use the start profile. A finite-difference
test in the own-reach-included mode checks the analytic gradient of the
doubled-proxy objective against central differences, to relative error
$<10^{-5}$. No parameter is learned: the decorrelation scale is drawn from
the fixed candidate set $\lambda\in\{0.1,0.3\}$, protocol configuration
shared by all reported runs.

\paragraph{Added solvers.}
The nine solvers of the second stage are local regret minimizers plugged into
the same engine; the certificate accumulators read the instantaneous regret
at the profile actually played and never the optimizer's discounted or
projected state, so Proposition~\ref{prop:sequence} applies unchanged. DCFR
discounts the accumulated regret by $t^\alpha/(t^\alpha+1)$ for positive and
$t^\beta/(t^\beta+1)$ for negative entries with $(\alpha,\beta)=(1.5,0)$;
DCFR$^+$ applies the positive discount and the $+$ floor, and PDCFR$^+$ adds
the one-step prediction of Table~1 of \citet{xu2024minimizing}; both use the
corresponding power weights in the average ($p{=}2$, $4$ and $5$). HS-PCFR$^+$
and HS-DCFR follow the linear schedule of equation~(4) of
\citet{zhang2026faster} with the total iteration count $n=T$ substituted,
$\gamma_s=\gamma_0-5s/n$ for the weight exponent (and the discount exponents
for HS-DCFR); the resulting weights reach $e^{369}$ at $T$ and stay inside
double precision, and the configuration check rejects schedules that would
not. Smooth PRM$^+$ projects the regret state onto $\{R\ge0,\lVert R\rVert_1
\ge 1\}$ per infoset, Stable PRM$^+$ restarts an infoset's state at the unit
level when it falls below it, and ExRM$^+$ takes the extragradient step of
Algorithm~5 of \citet{farina2023regret}; the step $\eta$ rescales the
regrets against that unit level, and because counterfactual regrets in
biased Leduc carry chance reach of about $1/30$ the grid was
$\eta\in\{10,10^2,10^3,10^4\}$, above the matrix-game values of the
paper. ExRM$^+$ exists only with simultaneous updates. IR-PRM$^+$ is
Algorithm~2 of \citet{zhang2025scale} with the previous utility vector as
prediction. Each family was run on seeds $0$--$9$ over a grid of $35$
configurations in total (discount exponents, schedules, $\eta$, averaging
exponent, update order), and one configuration per family was selected by a
rule written before the pilot: the highest oracle fraction at $10^{-5}$,
ties broken by the smaller median trajectory minimum. The selected
configurations are the nine solvers reported everywhere; seeds $0$--$9$ are
part of the in-sample set, so the sixteen-solver in-sample figures include this
selection step. The two held-out stages use the same $64$ instances
(seeds $64$--$127$): Stage~1 touched them first, and the second
preregistration was frozen before the Stage~2 runs, so the Stage~2 figures
reuse the Stage~1 instances and are not fresh data with respect to Stage~1.
Stage~3 uses a separate set (seeds $128$--$191$) that no earlier stage had
touched. Six added solvers ran on three-player
Leduc by a rule stated before the two-player results were read: the three
discounting-family solvers, plus every added solver whose two-player oracle
fraction reached $0.60$ (HS-PCFR$^+$, Smooth PRM$^+$, Stable PRM$^+$).

\paragraph{Portfolio and reported rows.}
The pool of $K=16$ solvers contains RM$^+$ and PRM$^+$ with linear ($p{=}1$)
and quadratic ($p{=}2$) weights, three DPRM$^+$ solvers that add a
clipped, scaled, tree-chained decorrelation direction from a
doubled-current-iterate proxy, and the nine added solvers above, all
training-free. ``Selection''
rows of Table~\ref{tab:main} report the exact NashGap of the iterate chosen
by a certificate. ``Certified'' rows report the fraction whose selected
certificate is itself below the threshold, the number a deployment can state
without an exact best-response pass against the average. The reference's
three-player benchmark is the fixed game of its Table~6.

\paragraph{Protocol.}
Biased Leduc: $\beta \sim U(0, \tfrac12)$ per instance drawn from
seeds $0$--$63$, paired across solvers, $T = 2^{18}$, our own dense grid of $128$ evenly spaced
checkpoints; the reference instead reports checkpoints at $2^k$, about
$19$ evaluations through this horizon. The three control solvers of
Table~\ref{tab:main} (uniform RM$^+$/PRM$^+$ and simultaneous PRM$^+$) use
the same paired seeds. Three-player Leduc: $T=2^{16}$, deterministic, with $64$ checkpoints spaced
every $1024$ iterations. Every run records $C^0$ and $C^1$ from the
accumulators of Proposition~\ref{prop:sequence}; the two-player grid and the
preselected three-player solver additionally log the aligned start-of-round
certificate (extra regret pass) for the comparison in
Section~\ref{sec:exp-alignment}. The aligned three-player rerun
(rerun in strict logging mode into a fresh output directory) reproduced
the original trajectory to floating-point precision; the solver was fixed
before it. All headline numbers are recomputed from the stored per-run
logs by the analysis scripts.

\section{Future Work}
\label{app:future-work}

The results point to the following experiments and technical extensions.

\paragraph{Reproducing the reported baselines.}
Our re-evaluation compares our runs against the \emph{reported} numbers
of \citet{sychrovsky2026approximating}: no code accompanies their paper,
which has no repository link or reproducibility statement and no public
release, so rerunning their solver is not possible. The paper does
not fully specify the averaging weights of the
plus-family baselines. We could not reproduce their PCFR$^+$ row under
any configuration we tried. Our uniform PRM$^+$ measures $0.00$
at $10^{-5}$ and $0.77$ at $10^{-4}$ against their reported $0.16$ and
$0.48$, mismatched in both directions, and our uniform-average runs remain
above $10^{-5}$ at this horizon for RM$^+$ and PRM$^+$. To test our own
rules and NashGap evaluator as the sources of this mismatch, we
cross-validated them against OpenSpiel \citep{lanctot2019openspiel}. A
reimplementation of biased Leduc written independently from the paper's
text matches OpenSpiel's built-in Leduc exactly at $\beta = 1$
(identical CFR$^+$ NashGap values at every logged checkpoint), and
injecting our exported average strategies into OpenSpiel's general-sum
gap computation reproduces our NashGap to relative error
$\le 3.9\times10^{-7}$ on all $12$ instance--solver pairs tested
(Appendix~\ref{app:openspiel}). These checks validate our own rules and
evaluator; the reference solver's configuration remains an open
reproduction target, and obtaining it would enable a controlled rerun.

\paragraph{A paired low-$\beta$ comparison.}
At threshold $10^{-5}$ on biased Leduc, NPCFR$^+$'s reported $0.96$
exceeds our certificate-selected $0.89$, and the two sides use
independently drawn instance samples. Our portfolio's $7$ in-sample failures
lean towards low $\beta$ and its $5$ held-out failures do not. A shared
instance set would test whether the reported aggregate margin concentrates
in the low-$\beta$ region. The NPCFR instance draws are unpublished, so the
comparison currently rests on aggregates.

\paragraph{Aligned logging as an ablation.}
Proposition~\ref{prop:sequence} makes the second regret pass against the
start-of-round profile unnecessary, but we paid it across the full
two-player biased-Leduc portfolio to compare the two certificates on
identical trajectories. Both are sound at all $131{,}072$ checkpoints; on
the seven original solvers the aligned certificate selects $0.80$ and certifies
$0.69$ at $10^{-5}$ against $0.84$ and $0.75$ for $C^1$, and on the
sixteen-solver pool it certifies $0.86$ against $0.89$
(Section~\ref{sec:exp-alignment}). The
comparison isolates a design principle. The certificate should follow the
regrets the algorithm actually minimizes and account for the resulting
profile mismatch explicitly. Forcing every accumulator onto one profile
costs tightness. The certified-anchor pattern also transfers beyond
regret-matching loops. In a side experiment on Leduc, a regularized
entropy-descent solver advancing a fixed-average reference maintained an
aligned certificate next to its working strategy, and all $13$ anchor
updates over $5{,}000$ iterations were exactly re-audited as sound (final
anchor NashConv $0.0139$ against its stored bound $0.0205$; the per-move
audit will be released with the code). That
deployment certifies the sum-form NashConv, so the per-player bounds
$\Rhat_i + \Delta_i$ enter as a sum
(Appendix~\ref{app:openspiel}).

\paragraph{Certificates for selection and for control.}
Using the certificate to \emph{control} the solver is a different regime
with a known failure mode. In the same side experiment, gating reference
advances on certified progress of a single sequence froze the dynamics
after four advances (final NashConv $0.109$ versus $0.022$ for the
ungated schedule). Forcing advances past the gate violated the gate's own
monotone-threshold invariant at every forced step. The stored acceptance
threshold no longer bounded the reference's NashConv, because an anchor
advance replaces the reference that bound was evaluated on. All $18$
forced advances in the pilot run violated the stored bound. Throughout,
the certificate
inequality $\sum_i(\Rhat_i+\Delta_i) \ge \mathrm{NashConv}$ held for the
anchor actually in force, whose audit logs will be released with the code.
These results separate certificate-gated control from certificate-based
selection; our protocol keeps the certificate out of the update loop.

\paragraph{Tighter multi-player certificates.}
On three-player Leduc the $M D_{-i}$ form of the certificate is
$12$--$387\times$ loose at $T$ across the six original solvers, $C^0$ is
$2.9$--$105\times$, and the exact comparator $C^1$ is $1.4$--$8.2\times$
(Table~\ref{tab:leduc3}); on each of them the comparator relaxation
$M D_{-i} - L_{-i}$ is essentially the whole slack of the $M D$ form
(median share $0.98$--$1.0$; Figure~\ref{fig:coalition} shows the PRM$^+$,
$p{=}2$ solver). The added HS-PCFR$^+$ solver is the exception that confirms the
attribution: its averaging weights sit almost entirely on the last iterates,
the coalition imbalance is near zero, and the three forms coincide at
$1.3\times$. Closing the remaining gap between $C^0$ and $C^1$ therefore requires
a zero-best-response surrogate for the comparator optimum $L_{-i}$, since
evaluating $L_{-i}$ costs a backward pass per player per checkpoint.
One-shot evaluation remains the general final-verification tier when an
exact guarantee is needed.

\paragraph{Broader benchmark coverage.}
The benchmarks come from the paper we re-evaluate and are
Leduc variants, deliberately so because the target claim concerns those
benchmarks. Appendix~\ref{app:external} adds two structurally different
general-sum EFGs, where selection and certification carry over. Sheriff is near-saturated for the predictive solvers, Goofspiel leaves
the DPRM$^+$ certificates visibly loose, and all games remain small
tabular benchmarks. The three-player result is a single deterministic
run, as in the reference protocol. Larger and more varied games are the
next test of mechanism transfer; so are payoffs other than chips, such as
tournament prize equity \citep{li2026icm}. Agents that maintain no regret
accumulator, such as LLM poker agents grounded in a rule library
\citep{li2026pokerskill}, fall outside the certificate's scope.

\paragraph{Sampled and amortized certificates.}
The certificate consumes exact weighted terminal distributions, an
$O(|\cZ|)$ object per checkpoint, and is therefore incompatible with
Monte Carlo CFR-style sampling \citep{lanctot2009monte} as it stands.
Relatedly, the DPRM$^+$ solver's $O(|\cZ|)$ decorrelation gradient per
iteration was infeasible at full length on three-player Leduc
($\approx 3.5$\,s per iteration), so the three-player portfolio ran
without DPRM$^+$ solvers. Sampled estimators or amortized terminal
accumulators are the natural next step.

\end{document}